\documentclass[preprint, 12pt]{article}

\usepackage[utf8]{inputenc}
\usepackage[T1]{fontenc}
\usepackage{lmodern}
\usepackage{microtype}
\usepackage{amsmath}
\usepackage{amsfonts}
\usepackage{amsthm}
\usepackage{graphicx}
\usepackage{hyperref}
\usepackage{ifthen}
\usepackage[ruled]{algorithm2e}
\usepackage[dvipsnames]{xcolor}
\usepackage{xspace}
\usepackage{array,multirow,dsfont}
    
\newtheorem{theorem}{Theorem}
\newtheorem{lemma}{Lemma}
\newtheorem{remark}{Remark}
\newtheorem{corollary}{Corollary}
\newcommand{\hei}[2][]{%
  \ifthenelse { \equal {#1} {} }%
  {#2.h}%
  {#2.h^{#1}}}

\newcommand{\state}{\ensuremath{\texttt{state}}\xspace}
\newcommand{\stateIA}[2]{\ensuremath{\texttt{st}^{#2}_{#1}}\xspace}
\newcommand{\stateI}[1]{\ensuremath{\texttt{st}_{#1}}\xspace}

\newcommand{\algo}{\ensuremath{\texttt{algo}}\xspace}
\newcommand{\InitialAlg}{\ensuremath{{AlgI}}\xspace}

\newcommand{\durationIA}{\ensuremath{{B}}\xspace}
\newcommand{\lazyI}{\ensuremath{{f}}\xspace}
\newcommand{\Time}{\ensuremath{{T}}\xspace}

\begin{document}

\begin{center}
  {\Large Making local algorithms efficiently self-stabilizing in
    arbitrary asynchronous environments}\bigskip\bigskip

  Stéphane Devismes\smallskip

  \emph{\footnotesize Laboratoire MIS, Université de Picardie,\\
    33 rue Saint Leu - 80039 Amiens cedex 1, France}\bigskip

  \normalsize David Ilcinkas, Colette Johnen, Frédéric Mazoit\smallskip

  \emph{\footnotesize Univ.  Bordeaux, CNRS, Bordeaux INP, LaBRI, UMR 5800,
    F-33400 Talence, France}

  \let\thefootnote\relax\footnotetext{Email Adresses:
    \url{stephane.devismes@u-picardie.fr} (Stéphane Devismes),
    \url{david.ilcinkas@labri.fr} (David Ilcinkas),
    \url{johnen@labri.fr} (Colette Johnen),
    \url{frederic.mazoit@labri.fr} (Frédéric Mazoit)}
\end{center}

\subsection*{Abstract}
This paper deals with the trade-off between time, workload, and
versatility in self-stabilization, a general and lightweight
fault-tolerant concept in distributed computing.

In this context, we propose a transformer that provides an
asynchronous silent self-stabilizing version $Trans(\InitialAlg)$ of
any terminating synchronous algorithm $\InitialAlg$.  The transformed
algorithm $Trans(\InitialAlg)$ works under the distributed unfair
daemon and is efficient both in moves and rounds.

Our transformer allows to easily obtain fully-polynomial silent
self-sta\-bi\-li\-zing solutions that are also asymptotically optimal
in rounds.

We illustrate the efficiency and versatility of our transformer with
several efficient (i.e., fully-polynomial) silent self-stabilizing
instances solving major distributed computing problems, namely vertex
coloring, Breadth-First Search (BFS) spanning tree construction,
$k$-clustering, and leader election.

\section{Introduction}

Fault tolerance is a main concern in distributed computing, but is
often hard to achieve; see, e.g.,~\cite{FLP85}. Furthermore, when it
can be achieved, it often comes at the the price of sacrificing
efficiency (in time, space, or workload) or versatility; see,
e.g.,~\cite{DR82c,CYZG14j}.  In this paper, we tackle the trade-off
between time, workload, and versatility in
self-stabilization~\cite{Di74}, a general and lightweight
fault-tolerant concept in distributed
computing~\cite{AlDeDuPe19}. Precisely, we consider a specialization
of self-stabilization called \emph{silent
  self-stabilization}~\cite{DoGoSc96}.

Starting from an arbitrary configuration, a self-stabilizing algorithm
enables the system to recover within finite time a so-called
legitimate configuration from which it satisfies an intended
specification.  Regardless its initial configuration, a silent
self-stabilizing algorithm~\cite{DoGoSc96} reaches within finite time
a configuration from which the values of the communication registers
used by the algorithm remain fixed.  Notice that silent
self-stabilization is particularly suited for solving \emph{static
  problems}\footnote{As opposed to \emph{dynamic problems} such as
  token circulation, a static problem defines a task of calculating a
  function that depends on the system in which it is
  evaluated~\cite{Ti06}.} such as leader election, coloring, or
spanning tree constructions. Moreover, as noted in~\cite{DoGoSc96},
silence is a desirable property. For example, it usually implies more
simplicity in the algorithmic design since silent algorithms can be
easily composed with other algorithms to solve more complex
tasks~\cite{AlDeDuPe19}.

Since the arbitrary initial configuration of a self-stabilizing system
can be seen as the result of a finite number of transient
faults,\footnote{A transient fault occurs at an unpredictable time,
  but does not result in a permanent hardware damage. Moreover, as
  opposed to intermittent faults, the frequency of transient faults is
  considered to be low.} self-stabilization is commonly considered as
a general approach for tolerating such faults in a distributed
system~\cite{Do00,AlDeDuPe19}. Indeed, self-stabilization makes no
hypotheses on the nature (e.g., memory corruption or topological
changes) or extent of transient faults that could hit the system, and
a self-stabilizing system recovers from the effects of those faults in
a unified manner.

However, such versatility comes at a price, e.g., after transient
faults cease, there is a finite period of time, called the
\emph{stabilization phase}, during which the safety properties of the
system are violated. Hence, self-stabilizing algorithms are mainly
compared according to their \emph{stabilization time}, the worst-case
duration of the stabilization phase.

In the distributed computing community, the correctness and efficiency
of algorithms is usually established by paper-and-pencil proofs. Such
proofs require the formal definition of a computational model for
which the distributed algorithm is dedicated. The atomic-state
model~\cite{Di74} is the most commonly used model in the
self-stabilizing area.  This model is actually a shared memory model
with composite atomicity: the state of each node is stored into
registers and these registers can be directly read by neighboring
nodes; moreover, in one atomic step, a node can read its state and
that of its neighbors, perform some local computations, and update its
state.  Hence, executions in this model proceed in atomic steps in
which some enabled nodes move, i.e., modify their local state.  The
asynchrony of the system is materialized by the notion of
\emph{daemon} that restricts the set of possible executions.  The
weakest (i.e., the most general) daemon is the \emph{distributed
  unfair daemon}. Hence, solutions stabilizing under such an
assumption are highly desirable, because they work under any daemon
assumption.

The stabilization time of self-stabilizing algorithms is usually
evaluated in terms of rounds, which capture the execution time
according to the speed of the slowest nodes. However, another crucial
issue is the number of local state updates, i.e., the number of
\emph{moves}. By definition, the stabilization time in moves captures
the total amount of computations an algorithm needs in order to
recover a correct behavior. Hence, the move complexity is rather a
measure of work than a measure of time. Now, minimizing the number of
state modifications allows the algorithm to use less communication
operations and communication bandwidth. As explained
in~\cite{DoGoSc96,AlDeDuPe19}, to implement an atomic-state model
solution in message passing, all nodes should permanently check
whether or not they should change their local state due to the
modification of some neighbors' local states. Now, instead of
regularly sending maybe heavy messages containing the full node local
state, one can adopt the lightweight approach proposed
in~\cite{DoGoSc96}: nodes regularly send a proof that the value of
their state has not changed; then, a node requests the full local
state of a neighbor only after a received proof shows a state
modification. Such proof can be very small compared to the real local
state: the node can just randomly generate a nonce\footnote{Nonce
  stands for ``number once''.}, salt the hash of its local state with
this nonce, and finally send a ``small'' message containing the hash
value together with the nonce. In this way, the number of ``heavy''
messages (i.e., those containing the current local state of some node)
depends on the number of moves, and the minimization of the number of
moves permits to drastically reduced the communication bandwidth
usage. This is especially true when the self-stabilizing solution is
silent~\cite{DoGoSc96}, as in this case there is no more moves after
the legitimacy is reached, and from that point, only small proofs are
regularly sent.

Until now, in the atomic-state model, techniques to design an
algorithm achieving a stabilization time polynomial in moves usually
made its rounds complexity inherently linear in $n$, the number of
nodes; see, e.g.,~\cite{CoDeVi09,AlCoDeDuPe17,DeIlJo22}. Now, the
classical nontrivial lower bounds for many problems is $\Omega(D)$
rounds~\cite{GeTi02}, where $D$ is the network diameter.  Moreover, in
many large-scale networks, the diameter is rather logarithmic on $n$,
so finding solutions achieving round complexities linear in $D$ is
more desirable. Yet, only a few asynchronous solutions (still in the
atomic-state model) achieving such an upper bound exist, e.g., the
atomic-state version of the Dolev's BFS algorithm~\cite{do93} given
in~\cite{DeJo16}. Moreover, it has been shown that several of them
actually have a worst-case execution that is exponential in moves;
see~\cite{DeJo16}. So, it seems that efficiency in rounds and moves
are often incompatible goals. In a best-effort spirit, Cournier et
al.~\cite{CoRoVi19} have proposed to study what they call
\emph{fully-polynomial} stabilizing solutions, i.e., stabilizing
algorithms whose round complexity is polynomial on the network
diameter and move complexity is polynomial on the network
size.\footnote{Actually, in~\cite{CoRoVi19}, authors consider atomic
  steps instead of moves. However, these two time units essentially
  measure the same thing: the workload. By the way, the number of
  moves and the number of atomic steps are closely related: if an
  execution $e$ contains $x$ steps, then the number $y$ of moves in
  $e$ satisfies $x \leq y \leq n\cdot x$.} As an illustrative example,
they have proposed a silent self-stabilizing BFS spanning tree
algorithm that stabilizes in $O(n^6)$ moves and $O(D^2)$ rounds in a
rooted connected network. Until now, it was the only asynchronous
self-stabilizing algorithm of the literature achieving this property.

\subsection{Contribution}
In this paper, we address the issue of generalizing the
fully-polynomial approach in the atomic-state model to obtain silent
self-stabilizing solutions that are efficient both in rounds and
moves.

To that goal, we propose to exploit the links, highlighted
in~\cite{LeSuWa09}, between the \emph{Local model}~\cite{Li92} and
self-stabilization in order to design an efficient \emph{transformer},
i.e., a meta-algorithm that transforms an input algorithm that does
not achieve a desired property (here, self-stabilization) into an
algorithm achieving that property.

Our transformer provides an asynchronous silent self-stabilizing
version $Trans(\InitialAlg)$ of any terminating synchronous algorithm
$\InitialAlg$ which works under the distributed unfair daemon and is
efficient both in moves and rounds. Precisely, our transformer has
several inputs: the algorithm $\InitialAlg$ to transform, a flag $f$
indicating the used transformation mode, and optionally a bound
$\durationIA$ on the execution time of $\InitialAlg$.

We have two modes for the transformation depending on whether the
transformation is \emph{lazy} or \emph{greedy}. In both modes, the
number of moves of $Trans(\InitialAlg)$ to reach a terminal
configuration is polynomial, in $n$ and the synchronous execution time
$\Time$ of $\InitialAlg$ in the lazy mode, in $n$ and $\durationIA$
otherwise.  An overview of $Trans(\InitialAlg)$ properties (memory
requirement, convergence in rounds, convergence in moves) according to
the both parameters is presented in Table~\ref{table:overview}.

In the lazy mode, the output algorithm $Trans(\InitialAlg)$ stabilizes
in $O(D+\Time)$ rounds where $\Time$ is the actual execution time of
$\InitialAlg$. Moreover, if the upper bound $\durationIA$ is given,
then the memory requirement of $Trans(\InitialAlg)$ is bounded
(precisely, we obtain a memory requirement in $O(\durationIA\times M)$
bits per node, where $M$ is the memory requirement of $\InitialAlg$) .

In the greedy mode, $Trans(\InitialAlg)$ stabilizes in
$O(\durationIA)$ rounds and its memory requirement is also bounded
(still $O(\durationIA\times M)$ bits per node).

Our transformer allows to drastically simplify the design of
self-stabilizing solutions since it reduces the initial problem to the
implementation of an algorithm just working in synchronous settings
with a pre-defined initial configuration. Moreover, this simplicity
does not come at the price of sacrificing efficiency since it allows
to easily implement fully-polynomial solutions.

Finally, our method is versatile, because compatible with most of
distributed computing models. Indeed, except when the computation of a
state of the input algorithm requires it, our transformer does not use
node identifiers nor local port numbers. More precisely, each move is
made based on the state of the node and the set of the neighbors'
states (if several neighbors have the same state $s$, the number of
occurrences of $s$ is not used to manage the simulation). Therefore,
our transformer can be used in strong models with node identifiers
like the \emph{Local} model~\cite{Li92}, down to models almost as weak
as the \emph{stone age} model~\cite{EW13c}.

Our solution is very efficient in terms of time and workload, but at
the price of multiplying the memory cost of the original algorithm by
its execution time. This time (and thus the multiplicative factor) is
however usually small in powerful models such as the Local model. As
an illustrative example, this multiplicative memory overhead can be as
low as $O(\log^* n)$ for very fast algorithms such as the Cole and
Vishkin's coloring algorithm~\cite{CoVi86}. Furthermore, using nonces
and hashes similarly as explained before, one can reduce the
communication cost to almost the one of the original
algorithm. Indeed, our transformer always modifies its state in a way
which can be described with limited information, linear in the time
and memory complexity of the simulated algorithm.

We illustrate the efficiency and versatility of our proposal with
several efficient (i.e., fully-polynomial) silent self-stabilizing
solutions for major distributed computing problems, namely vertex
coloring, Breadth-First Search (BFS) spanning tree construction,
$k$-clustering, and leader election. In particular, we positively
answer to some open questions raised in the conclusion
of~\cite{CoRoVi19}: (1) there exists a fully-polynomial (silent)
self-stabilizing solution for the BFS spanning tree construction whose
stabilization time in rounds is asymptotically linear in $D$ (and so
asymptotically optimal in rounds), and (2) there exists a
fully-polynomial (silent) self-stabilizing solution for the leader
election (also with a stabilization time in rounds that is linear in
$D$ and so asymptotically optimal in terms of rounds). Finally, we can
also design for the first time (to the best of our knowledge)
asynchronous fully-polynomial self-stabilizing algorithms with a
stabilization time in rounds that can be sublinear in $D$, as shown
with our vertex coloring instance that can stabilize in $O(\log^* n)$
rounds in unidirectional rings using the right parameters.

\renewcommand{\arraystretch}{1.5}
\begin{table}[htbp]
  \centering

  \begin{tabular}{c|c|c|}
    \multicolumn{1}{c}{}&\multicolumn{1}{c}{Move complexity} & \multicolumn{1}{c}{Round complexity}\\
    \cline{2-3}
    Lazy mode& $O(\min(n^3+ nT, n^2\durationIA))$ & $O(D+T)$\\
    \cline{2-3}
    Greedy mode& $O(\min(n^3+n\durationIA, n^2\durationIA))$ & $O(B)$\\
    \cline{2-3}
    \noalign{\bigskip Common features\medskip}
    \cline{2-3}
    Error recovery& $O(\min(n^3, n^2\durationIA))$ & $O(\min(D, B))$\\
    \cline{2-3}
    \multicolumn{1}{c}{Space complexity} &\multicolumn{2}{c}{$\durationIA\cdot M$}\\
  \end{tabular}

  \small
  \begin{itemize}
  \item $T$ and $M$ are respectively the time and space complexities
    of $\InitialAlg$.  \vspace*{-3pt}
  \item $B$ is a parameter $\in \mathbb{N}\cup\{+\infty\}$.
    \vspace*{-3pt}
  \item The algorithm is always silent when $\durationIA<+\infty$.
    Here, we assume that $T\leq \durationIA$.
  \end{itemize}
  \vspace*{-12pt}
  \caption{Overview of the properties of $Trans(\InitialAlg)$.}
  \label{table:overview}
\end{table}

\subsection{Related Work}
Proposing \emph{transformers} (also called \emph{compilers}) is a very
popular generic approach in self-stabilization.  Transformers are
useful to establish expressiveness of a given property: by giving a
general construction, they allow to exhibit a class of problems that
can achieve a given property. An impossibility proof should be then
proposed to show that the property is not achievable out of the class,
giving thus a full characterization.  For example, Katz and
Perry~\cite{KaPe93} have addressed the expressiveness of
self-stabilization in message-passing systems where links are reliable
and have unbounded capacity, and nodes are both identified and
equipped of infinite local memories.  Several transformers,
e.g.,~\cite{BoVi02,KaPe93}, builds time-efficient self-stabilizing
solutions yet working synchronous systems only.  Bold and
Vigna~\cite{BoVi02} proposes a universal transformer for synchronous
networks. As in~\cite{KaPe93}, the transformer allows to
self-stabilize any behavior for which there exists a self-stabilizing
solution. The produced output algorithm stabilizes in at most $n + D$
rounds. However, the transformer is costly in terms of local memories
(basically, each node collects and stores information about the whole
network). In the same vein, Afek and Dolev~\cite{AfDo02} propose to
collect pyramids of views of the system to detect incoherences and
correct the behavior of a synchronous system.  In~\cite{BlFrPa14},
Blin et al. propose two transformers in the atomic-state model to
construct silent self-stabilizing algorithms. The first one aims at
optimizing space complexity: if a task has a proof-labeling scheme
that uses $\ell$ bits, then the output algorithm computes the task in
a silent and self-stabilizing manner using $O(\ell + \log n)$ bits per
node.  However, for some instances, it requires an exponential number
of rounds.  The second one guarantees a stabilization time in $O(n)$
rounds using $O(n^2+k\cdot n)$ bits per node for every task that uses
a $k$-bit output at each node.  Transformers have been also used to
compare expressiveness of computational models. Equivalence (in terms
of computational power) between the atomic-state model and the
register one and between the register model and message passing are
discussed in~\cite{Do00}. In~\cite{Tu12}, Turau proposes a general
transformation procedure that allows to emulate any algorithm for the
distance-two atomic-state model in the (classical) distance-one
atomic-state model assuming that nodes have unique identifiers.

It is important to note that the versatility is often obtained at the
price of inefficiency: the aforementioned transformers use heavy (in
terms of memory and/or time) mechanisms such as global snapshots and
resets in order to be very generic. For example, the transformer
proposed by Turau~\cite{Tu12} increases the move complexity of the
input algorithm by a multiplicative factor of $O(m)$ where $m$ is the
number of links in the network. The transformer of Katz and
Perry~\cite{KaPe93} requires infinite local memories and endlessly
computes (costly) snapshots of the network even after the
stabilization. Lighter transformers have been proposed but at the
price of reducing the class of problems they can handle. For example,
locally checkable problems are considered in the message-passing
model~\cite{AfKuYu97}. The proposed transformer constructs solutions
that stabilize in $O(n^2)$ rounds. The more restrictive class of
locally checkable and locally correctable problems is studied
in~\cite{AwPaVa91}, still in message-passing.  Cohen et
al.~\cite{CPRS23} propose to transform synchronous distributed
algorithms that solve locally greedy and mendable problems into
asynchronous self-stabilizing algorithms in anonymous
networks. However, the transformed algorithm requires a strong
fairness assumption called the Gouda fairness in their paper. This
property is also known as the strong global fairness in the
literature. Under such an assumption, move complexity cannot be
bounded in general.  Finally, assuming the knowledge of the network
diameter, Awerbuch and Varghese~\cite{AwVa91} propose, in the
message-passing model, to transform synchronous terminating algorithms
into self-stabilizing asynchronous algorithms. They propose two
methods: the \emph{rollback} and the \emph{resynchronizer}. The
resynchronizer additionally requires the input algorithm to be locally
checkable. Using the rollback (resp., resynchronizer) method, the
output algorithm stabilizes in $O(T)$ rounds (resp., $O(T+U)$ rounds)
using $O(T \times S)$ space (resp., $O(S)$ space) per node where $U$
is an upper bound on the network diameter and $T$ (resp., $S$) is the
execution time (resp., the space complexity) of the input algorithm.
Notice however that the straightforward atomic-state model version of
the rollback compiler (the work closest to our contribution) achieves
exponential move complexities, as shown in Section~\ref{sect:roll}.

\subsection{Roadmap}
The rest of the paper is organized as follows. In the next section, we
define the model of the input algorithm, the property the input
algorithm should fulfill, and the output model. In
Section~\ref{sect:trans}, we present our transformer. In
Sections~\ref{sect:proof:deb}-\ref{sect:round}, we establish the
correctness and the complexity of our method. In
Section~\ref{instances}, we illustrate the versatility and the power
of our approach by proposing several efficient instances solving
various benchmark problems. In Section~\ref{sect:roll}, we establish
an exponential lower bound in moves for the straightforward
atomic-state model version of the rollback compiler of Awerbuch and
Varghese~\cite{AwVa91} (the closest related work). We make concluding
remarks in Section~\ref{ccl}.

\section{Preliminaries}

\subsection{Networks}

We consider \emph{distributed systems} made of $n \geq 1$
interconnected nodes. Each node can directly communicate through
channels with a subset of other nodes, called its neighbors.  We
assume that the network is connected\footnote{If the network is not
  connected, then the algorithm runs independently in each connected
  component and its analysis holds nevertheless.} and that
communication is bidirectional.

More formally, we model the topology by a symmetric strongly-connected
simple directed graph $G = (V, E)$, where $V$ is the set of nodes and
$E$ is the set of arcs, representing the communication channels.  Each
arc $e$ goes \emph{from a node $p$ to a node $q$}, and we respectively
call $p$ and $q$ the \emph{source} and \emph{destination} of $e$.  The
graph is \emph{symmetric}, meaning that for every arc from $p$ to $q$,
there exists an arc from $q$ to $p$.  We denote by $N(p)$ the set of
nodes such that there exists an arc from $p$ to $q$.  The elements of
$N(p)$ are the \emph{neighbors} of $p$.  We denote by $C(p)$ the set
of incoming arcs of node $p$. For any $c\in C(p)$, we denote by $q_c$
the neighbor of $p$ at the opposite side of the arc.

A \emph{path} is a finite sequence $P=p_0p_1\cdots p_l$ of nodes such
that consecutive nodes in $P$ are neighbors. We say that $P$ is
\emph{from} $p_0$ \emph{to} $p_l$.  The \emph{length} of the path $P$
is the number $l$.  Since we assume that $G$ is \emph{connected}, then
for every pair of nodes $p$ and $q$, there exists a path from $p$ to
$q$.  We can thus define the \emph{distance} between two nodes $p$ and
$q$ to be the minimum length of a path from $p$ to $q$.  The
\emph{diameter} $D$ of $G$ is the maximum distance between nodes of
$G$.  Given a non-negative integer $k$, the \emph{ball of radius $k$
  around a node $p$} is the set $N^k[p]$ of nodes at distance at most
$k$ from $p$.  The \emph{closed neighborhood} of a node $p$ is the set
$N[p]= N^1[p]$.  Note that for every node $p$, $N(p)=N[p]-\{p\}$ and
$N^D[p]=V$.

\subsection{Input Computational Model: Eventually Stable Distributed
  Synchronous Algorithms}

Since we want to transform algorithms operating on various settings,
we first define a general model that suits all these settings.

In this paper, we define a \emph{distributed synchronous algorithm} by
the following four elements:
\begin{itemize}
\item a datatype \state to label nodes;
\item a datatype \texttt{label} to label communication channels (i.e.,
  arcs)); this datatype can be reduced to a singleton $\{\bot\}$;

\item a computational function \algo that returns a state of type
  \state given a state (in \state) and a set of pairs (label, state)
  of type \texttt{label} $\times$ \state;

\item a predicate \texttt{isValid} which takes as input a labeled
  graph and returns $\texttt{true}$ if and only if the labeled graph
  constitutes a valid initial configuration (this predicate a priori
  depends on the model and the problem).
\end{itemize}

Let us now describe the execution model for distributed synchronous
algorithms.  First, each channel of the network has a fixed label in
\texttt{label} and each node has a pre-defined initial state in
\state, thus giving a labeled graph, and this graph satisfies the
\texttt{isValid} predicate.

Then, executions proceed in synchronous rounds where every node (1)
obtains information from all its neighbors by the mean of its incoming
channels and (2) computes its new state accordingly using the function
\algo.

At each round, $p$ obtains a pair
$(\texttt{lbl}_c, \stateIA{q_c}{\,})$ from each channel $c\in C(p)$,
where $\texttt{lbl}_c$ is the label of $c$ and $\stateIA{q_c}{\,}$ is
the current state of $q_c$.  Hence, to compute its new state, $p$
knows its own state and the set
$\{(\texttt{lbl}_c, \stateIA{q_c}{\,}) \mid c\in C(p)\}$, which
constitute the inputs of \algo.  Notice that, in the case when
\texttt{label} is $\{\bot\}$, \algo computes a new state according to
the current state of the node and the set
$\{(\bot, \stateIA{q_c}{\,}) \mid c\in C(p)\}$.  Hence, depending on
the values of channel labels, $p$ may or may not be able to locally
identify its neighbors, or to have data about the channels such as
their bandwidth, their cost, \ldots

The state of a node $p$ at the beginning of the round $i$ is denoted
$\stateIA{p}{i}$.  Thus, at each round $i$, each node $p$ computes its
new state $\stateIA{p}{i+1}$ for the next round as
$\stateIA{p}{i+1} = \algo( \stateIA{p}{i}, \{(\texttt{lbl}_c,
\stateIA{q_c}{i}) \mid c\in C(p)\})$.

We say that a distributed synchronous algorithm is \emph{eventually
  stable} if, for any labeled graph satisfying \texttt{isValid}, there
exists a round number $s$ such that, for any node $p$ and any
$i \geq s$, $\stateIA{p}{i+1} = \stateIA{p}{i}$.  Note that if all
states remain stable in a round, then they remain stable forever when
\algo is deterministic.

\paragraph{Accommodating various cases}

As already stated, our model is deliberately general in order to
accommodate most standard distributed computing models.

For example, if one considers identified networks, the $\state$ of
each node can contain an identifier, and \texttt{isValid} can check
that these identifiers are different for each node.  It may also check
that these identifiers are not too large with respect to the size $n$
of the network, in cases where we assume that the identifiers use a
number of bits polylogarithmic in $n$.  Note that in the case of
identified networks, assigning labels to channels does not add
anything useful to the model, and thus $\texttt{label}$ can be the
singleton $\{\perp\}$.

If we do not need the full power of identified networks but only need
that a node can distinguish its neighbors, \texttt{isValid} can check
that, locally, channels have distinct labels.

In semi-uniform networks with a distinguished node (usually called the
root), a boolean may be used in $\state$ to designate this special
node.  The fact that exactly one node is selected is checked by the
\texttt{isValid} predicate.

If the synchronous algorithm that we want to simulate is explicitly
terminating, we can turn it into an eventually stable distributed
synchronous algorithm by keeping the same state forever instead of
terminating.

Note that our formalism assumes that the whole state eventually
stabilizes.  We made this choice to keep things simple.  However, if
this requirement is too strong, our results can be easily extended to
the case when there is a function \texttt {rst} from $\state$ to some
other datatype $\texttt{state'}$, and only the result of this function
is eventually stable (i.e.,
$\texttt {rst}(\stateIA{p}{i+1}) =\texttt {rst}(\stateIA{p}{i})$, for
all sufficiently large $i$).

To illustrate the versatility of our model and the power of our
approach, we will propose several examples of eventually stable
distributed synchronous algorithms solving benchmark problems in
Section~\ref{instances}, page~\pageref{instances}.

\subsection{Output Computational Model: the Atomic-state Model}

Instances of our transformer run on the \emph{atomic-state
  model}~\cite{AlDeDuPe19} in which nodes communicate using a finite
number of locally shared registers, called \emph{variables}. Some of
these shared registers may be read-only: they cannot be modified by
the algorithm nor by faults.  This is typically the case for unique
identifiers, if some are available.  In one indivisible move, each
node reads its own variables and those of its neighbors, performs
local computation, and may change only its own variables.  The
\textit{state} of a node is defined by the values of its local
variables.  A \emph{configuration} of the system is a vector
consisting of the states of each node.

To accommodate the same diversity of models as in the input
computational model, we also describe an output algorithm by four
elements:
\begin{itemize}
\item a datatype \state to label nodes;
\item a datatype \texttt{label} to label communication channels (i.e.,
  arcs)); this datatype can be reduced to a singleton $\{\bot\}$;

\item a computational function \algo that returns a state $s'$ of type
  \state given a state $s$ (in \state) and a set of pairs (label,
  state) of type \texttt{label} $\times$ \state; in this case, the
  read-only registers must have the same value in $s$ and in $s'$;

\item a predicate \texttt{isValid} which takes as input a labeled
  graph; in this case, the nodes are only labeled with the read-only
  shared registers; this predicate a priori depends on the model and
  the problem.
\end{itemize}

The \textit{program} \algo of each node is described as a finite set
of \textit{rules} of the form $label:guard \to\ action$.
\emph{Labels} are only used to identify rules in the reasoning.  A
\textit{guard} is a Boolean predicate involving the state of the node
and that of its neighbors.  The \emph{action} part of a rule updates
the state of the node.  A rule can be executed only if its guard
evaluates to \emph{true}; in this case, the rule is said to be
\emph{enabled}.  By extension, a node is said to be enabled if at
least one of its rules is enabled.  We denote by
$\textit{Enabled}(\gamma)$ the subset of nodes that are enabled in
configuration $\gamma$.

In the model, executions proceed as follows. When the configuration is
$\gamma$ and $\textit{Enabled}(\gamma) \neq \emptyset$, a non-empty
set $\mathcal X \subseteq \textit{Enabled}(\gamma)$ is selected by a
so-called \emph{daemon}; then every node of $\mathcal X$
\emph{atomically} executes one of its enabled rules, leading to a new
configuration~$\gamma^\prime$.  The atomic transition from~$\gamma$
to~$\gamma^\prime$ is called a \emph{step}. We also say that each node
of~$\mathcal X$ executes an \emph{action} or simply a \emph{move}
during the step from~$\gamma$ to $\gamma^\prime$.  The possible steps
induce a binary relation over $\mathcal C$, denoted by $\mapsto$.  An
\emph{execution\/} is a maximal sequence of configurations
$e=\gamma_0\gamma_1\cdots \gamma_i\cdots$ such that
$\gamma_{i-1}\mapsto\gamma_i$ for all $i>0$.  The term ``maximal''
means that the execution is either infinite, or ends at a
\emph{terminal} configuration in which no rule is enabled at any node.

As explained before, each step from a configuration to another is
driven by a daemon.  We define a daemon as a predicate over
executions.  We say that an execution $e$ is \emph{an execution under
  the daemon $S$} if $S(e)$ holds.  In this paper we assume that the
daemon is \emph{distributed} and \emph{unfair}.  ``Distributed'' means
that while the configuration is not terminal, the daemon should select
at least one enabled node, maybe more.  ``Unfair'' means that there is
no fairness constraint, i.e., the daemon might never select an enabled
node unless it is the only enabled node.  In other words, the
distributed unfair daemon corresponds to the predicate $true$, i.e.,
this is the most general daemon.

In the atomic-state model, an algorithm is \emph{silent} if all its
possible executions are finite.  Hence, we can define silent
self-stabilization as follows.  Let $\mathcal{L}$ be a non-empty
subset of configurations, called the set of legitimate configurations.
A distributed system is \emph{silent and self-stabili\-zing} under the
daemon $S$ for $\mathcal{L}$ if and only if the following two
conditions hold:
\begin{itemize}
\item all executions under $S$ are finite, and
\item all terminal configurations belong to $\mathcal{L}$.
\end{itemize}

We use two units of measurement to evaluate the time complexity:
\emph{ moves} and \emph{rounds}.  The definition of a round uses the
concept of \emph{neutralization}: a node $p$ is \textit{neutralized}
during a step $\gamma_i \mapsto \gamma_{i+1}$, if $p$ is enabled in
$\gamma_i$ but not in configuration $\gamma_{i+1}$, and does not
execute any action in the step $\gamma_i \mapsto \gamma_{i+1}$.  Then,
the rounds are inductively defined as follows.  The first round of an
execution~$e = \gamma_0\gamma_1\cdots$ is the minimal prefix~$e'$ such
that every node that is enabled in~$\gamma_0$ either executes a rule
or is neutralized during a step of~$e'$. If $e'$ is finite, then let
$e''$ be the suffix of $e$ that starts from the last configuration of
$e'$; the second round of~$e$ is the first round of $e''$, and so on
and so forth.

The \emph{stabilization time} of a silent self-stabilizing algorithm
is the maximum time (in moves or rounds) over every execution possible
under the considered daemon (starting from any initial configuration)
to reach a terminal (legitimate) configuration.

\section{Compile synchronous algorithms into self-stabilizing
  asynchronous ones}\label{sect:trans}

\subsection{Algorithm overview}

Before we give the actual algorithm, we give some general ideas on how
the algorithm operates.  We do this in hope that it helps the reader
get a better understanding of the algorithm before getting into the
proofs.  Note that some definitions that we give in this subsection
are not the ``real'' definitions.  Their purpose is only to clarify
how the algorithm works.

Our algorithm takes an eventually stable distributed synchronous
algorithm $\InitialAlg$, a flag $f \in \{lazy,greedy\}$, and possibly
a bound $B$ on its execution, as inputs.  Its output is a
self-stabilizing asynchronous algorithm which simulates $\InitialAlg$.

The first idea to turn $\InitialAlg$ into a self-stabilizing algorithm
is to store the whole execution of the algorithm.  Every node $p$ thus
has a list $L$ such that, ultimately, $p.L[i]=\stateIA{p}{i}$ for each
cell $i$.  We also denote $p.L[0]=\stateIA{p}{0}$, which cannot be
corrupted.  Since the cells $q.L[i]$ with $q\in N[p]$ constitute part
of the input that the algorithm uses to compute $p.L[i+1]$, we say
that $p.L[i+1]$ \emph{depends} on the cells $q.L[i]$.  We extend this
definition by taking its transitive closure.

A possible algorithm could be the following.  When $p$ is activated,
it finds its faulty cells, and corrects all of them.  Also if $p.L[i]$
does not exist but all its dependencies do, then it creates $p.L[i]$.
This gives an algorithm which is quite good, round-wise.  Indeed, at
the beginning, the cells $p.L[0]$ are valid by construction, and if
all $p.L[i]$ are valid, then after one round, so are all $p.L[i+1]$.
Thus if the synchronous algorithm finishes in $\Time$ steps, then its
simulation converges in $\Time$ rounds.

The main problem is that this algorithm may use an exponential number
of steps.  Indeed, such an algorithm tries to update the values in the
list as soon as possible, which may cause a lot of updates.  If
between two nodes $p$ and $q$, there exists two disjoint paths of
different lengths, an update on $p$ may trigger two updates of $q$.
However, these two updates of $q$ may trigger four updates on a
further node $r$, by the same construction and argument.  Repeating
this argument along a long chain of such gadgets leads to an
exponential number of updates (and thus of moves) caused by a single
original event.  Section~\ref{sect:roll} proves this in detail.

To avoid this kind of problem, our algorithm uses a more conservative
approach.  Whenever a node $p$ detects a ``major error'', it launches
a node which ensures that the buggy cell and all the cells which
depend on it are removed.  Only then can $p$ resume its computation.
We refer to this process as an \emph{error broadcast}.

To achieve this, we add a variable $p.s$ which can be either $C$ or
$E$ and such that $p.s=E$ means that $p$ has launched an error
broadcast.  This broadcast can either be an initial one or a
sub-broadcast launched to propagate an initial one.  Whenever $p$
knows that its broadcast is finished, it sets $p.s=C$.  In the
following, we denote the length of $p.L$ by $\hei{p}$, and note that
$p$ has a cell which has a missing dependency in $q.L$ if and only if
$\hei{p}\geq \hei{q}+2$.

Our algorithm has the following four rules:
\begin{itemize}
\item[~]\hspace*{-1em}\emph{Rule $R_R$:} Whenever $p$ encounters a
  major error, it applies the rule $R_R$, which empties the list $p.L$
  and sets $p.s=E$.  We explain what a major error is a bit later.

\item[~]\hspace*{-1em}\emph{Rule $R_P$:} If $p$ has a neighbor
  $q\in N(p)$ such that both $q.s=E$ and $\hei{p}\geq \hei{q}+2$, then
  $p$ should apply the rule $R_P$ to remove the problematic cells and
  propagate an error broadcast by setting $p.s$ to $E$.  A node can
  apply the rule $R_p$ as often as needed but, otherwise, it must wait
  for its error broadcast to finish.
\end{itemize}

We can now explain more precisely when $p$ has a major error, and we
do so before presenting the two other rules:
\begin{itemize}
\item $p$ has a cell $p.L[i]$ which has all its dependencies but the
  value $p.L[i]$ is incorrect.
\item $p.s=C$ and some neighbor $q\in N(p)$ is such that
  $\hei{q}\geq \hei{p}+2$.

  Note that the sole condition on the heights is not enough to create
  a serious problem.  Indeed, such a situation with $p.s = E$ is bound
  to happen during an error broadcast.

\item The last major error relates to the error state.  Indeed, a node
  $p$ should have two ways to set $p.s=E$.  Either $p$ has applied the
  rule $R_R$ and $\hei{p}=0$, or $p$ has applied the rule $R_P$ and it
  has a neighbor $q\in N(p)$ such that $q.s=E$ and $\hei{q}<\hei{p}$.
  If none of these conditions are met for a node $p$ such that
  $p.s = E$, then this is a serious error.
\end{itemize}

\begin{itemize}
\item[~]\hspace*{-1em}\emph{Rule $R_C$:} If $p$ knows that its error
  broadcast is finished, then it applies the rule $R_C$, which simply
  consists in switching $p.s$ from $E$ to $C$.  To give the precise
  conditions which allow the use of rule $R_C$, it is easier to see
  when $p$ should not apply this rule.  Indeed,
  \begin{itemize}
  \item if some neighbor $q\in N(p)$ is such that
    $|\hei{q}-\hei{p}|\geq 2$, then $p$ is involved in an error
    broadcast.  Indeed, if $\hei{p}\leq\hei{q}-2$, then $p$ must wait
    for $q$ to propagate its error broadcast.  Otherwise, if
    $\hei{p}\geq\hei{q}+2$ and $q.s=E$, then $p$ must propagate the
    error broadcast of $q$.  Finally, if $\hei{p}\geq\hei{q}+2$ and
    $q.s=C$, then $q$ has a major error and $p$ must wait for $q$ to
    apply the rule $R_R$ after which it will propagate the error
    broadcast of $q$.
  \item if some neighbor $q\in N(p)$ is such that $\hei{q}=\hei{p}+1$
    and $q.s=E$, then the broadcast of $p$ still concerns $q$ and thus
    is not finished, and $p$ must wait.
  \end{itemize}
  If neither conditions apply, then $p$ can apply rule $R_C$.

\item[~]\hspace*{-1em}\emph{Rule $R_U$:} To finish our overview of the
  algorithm, we should talk about the rule $R_U$, which performs the
  actual computation, and about the two functioning modes of our
  algorithm: greedy and lazy.

  Obviously, to apply the rule $R_U$ and create the new cell
  $p.L[\hei{p}+1]$, we should have $p.s=C$ and all the corresponding
  dependencies must exist, thus $\hei{q}\geq \hei{p}$ for any neighbor
  $q\in N(p)$ but $\hei{q}\leq \hei{p}+1$.

  \begin{itemize}
  \item In ``greedy mode'', whenever this condition is met, $p$ can
    apply the rule $R_U$.
  \item Now if $\stateIA{p}{\hei{p}+1}=\stateIA{p}{\hei{p}}$, it may
    be that the simulated algorithm has finished.  In lazy mode, by
    default, $p$ thus does not apply the rule $R_U$.  Nevertheless, if
    a neighbor $q\in N(p)$ is such that $\hei{q}>\hei{p}$, then it may
    be the sign that the computation may have locally converged but
    not globally, and $p$ does apply the rule $R_U$ in this case.
  \end{itemize}
\end{itemize}

\subsection{Data structures}

Let \InitialAlg be an eventually stable distributed synchronous
algorithm.  Let~$\Time$ be the number of synchronous rounds
\InitialAlg requires to become stable.

We now propose a transformer that takes this algorithm as input and
transforms it into an efficient fully asynchronous self-stabilizing
algorithm.  Our algorithm can run in two modes: lazy and greedy.  It
thus takes two additional parameters as inputs:
\begin{itemize}
\item a parameter \lazyI which can take 2 values: ``lazy'' or
  ``greedy''.
\item an upper bound $\durationIA$ on $\Time$, which can be set to
  $+\infty$ to simulate that such a bound is not available to the
  transformer.
\end{itemize}

The shared variables of the node $p$ are:
\begin{itemize}
\item $p.init$: an initial state of $p$ in the simulated algorithm; it
  cannot be modified and constitutes the read-only part of the state;
\item $p.s$: the status of $p$; it can take two values, $C$ or $E$;
\item $p.L$: a list of at most $B$ elements containing states of
  \InitialAlg.
\end{itemize}

The channels (arcs) of the network have the same labels as
in~\InitialAlg.  The predicate \texttt{isValid} is also the same
in~\InitialAlg and in the transformed algorithm.

A node $p$ such that $p.s = C$ is said to be \emph{correct}; otherwise
it is an \emph{erroneous} node (in other words, a node in error).  We
use the following useful notations and functions:
\begin{itemize}
\item We denote by $\hei[h]{p}$ the length of $p.L$ in $\gamma^h$.  If
  the configuration is clear from the context, we simply denote this
  length by $\hei{p}$.
\item We denote by $p.L[i]$ the element of $p.L$ at index $i$
  ($1\leq i\leq \hei{p}$).
\item Although it does not belong to $p.L$, we also denote by $p.L[0]$
  the initial value $p.init$ of $p$ for \InitialAlg (we also refer to
  $p.init$ as \stateIA{p}{0}).

\item $\hei{p}:=i$ is the truncation of $p.L$ at its first $i$
  elements.
\item $push(p, val)$ is the addition of the value $val$ at the end of
  the list $p.L$.
\end{itemize}

As already stated in the overview of the algorithm, ultimately, we
want $p.L[i]=\stateIA{p}{i}$.  We thus must be able to call the
simulated algorithm on the cells of $p.L$ and those of its neighbors.
To simplify notations, we set
\begin{eqnarray*}
  \widehat\algo(p, i)&:=&\algo\bigl(p.L[i], \{(\texttt{lbl}_c, q_c.L[i])\mid c\in C(p)\}\bigr)\\
  \exists q\in N(p), P(\texttt{st}(q)) &:=&
                                            \exists (\texttt{lbl}, \texttt{st})\in \{(\texttt{lbl}_c, \stateIA{q_c}{})\mid c\in C(p)\}, P(\texttt{st})\\
  \forall q\in N(p), P(\texttt{st}(q)) &:=&
                                            \forall (\texttt{lbl}, \texttt{st})\in \{(\texttt{lbl}_c, \stateIA{q_c}{})\mid c\in C(p)\}, P(\texttt{st})
\end{eqnarray*}
Recall that $C(p)$ and $q_c$ respectively denote the set of incoming
channels of $p$ and the source node of $c$.  The last two notations
are somewhat misleading because, although they may suggest it, they do
not rely on the fact that nodes can individually access their
neighbors.

If $\gamma^0\gamma^1\cdots$ is an execution, we respectively denote by
$p.s^i$, $p.L^i$, $\hei[i]{p}$ and $\widehat\algo(p^i, j)$ the value
of $p.s$, $p.L$, $\hei{p}$ and $\widehat\algo(p, j)$ in $\gamma^i$.

\subsection{The predicates}
\allowdisplaybreaks
\begin{eqnarray*}
  algoError(p)&:=&\exists i, 1 \leq i \leq \hei{p}, \; (\forall q\in N(p),
                   \hei{q}\geq i-1)\wedge\\
              &&\phantom{\exists i\leq \hei{p},\;(}
                 p.L[i] \neq \widehat\algo(p, i-1) \\
  \\
  dependencyError(p)&:=&\bigl(p.s=E\wedge \neg(\exists q\in N(p),\;
                         q.s=E\wedge \hei{q}<\hei{p})\bigr)\\
              &\vee&\bigl(p.s=C\wedge \exists q\in N(p),\; (\hei{q}\geq \hei{p}+2)\bigr)\\
  \\
  root(p)&:=&algoError(p)\vee dependencyError(p)\\
  \\
  errorPropag(p, i)&:=& \exists q\in N(p),\; q.s=E \wedge \hei{q}<i<\hei{p}\\
  \\
  canClearE(p)&:=& p.s=E\\
              &\wedge&\forall q \in N(p),\; \bigl(|\hei{q}-\hei{p}|\leq 1 \wedge\\
              &&\phantom{\forall q \in N(p),\; \bigl(}(\hei{q}\leq
                 \hei{p} \vee q.s=C)\bigr)\\
  \\
  updatable(p)&:=& p.s=C\quad \wedge\quad \hei{p} < \durationIA\quad\wedge\\
              && \bigl(\forall q \in N(p),\; \hei{q}\in\{\hei{p} , \hei{p}+1\}\bigr)\wedge\\
              && \bigl(\lazyI=\text{greedy} \vee (p.L[\hei{p}] \neq \widehat\algo(p, \hei{p})\vee\\
              &&\phantom{\bigl(\lazyI=\text{greedy} \vee (}\exists q\in N(p),\; \hei{q}>\hei{p})
                 \bigr)
\end{eqnarray*}

\subsection{The rules}

\begin{itemize}
\item
  $R_R: (\hei{p}>0\vee p.s=C) \wedge root(p) \longrightarrow \hei{p}
  := 0 \; ; \; p.s:=E$
\item
  $R_P(i): errorPropag(p,i)\longrightarrow \hei{p} := i \; ; \; p.s :=
  E$
\item $R_C: canClearE(p) \longrightarrow p.s := C$
\item
  $R_U: updatable(p) \longrightarrow push(p,\widehat\algo(p,
  \hei{p}))$
\end{itemize}
We set the following priorities:
\begin{itemize}
\item $R_R$ has the highest priority.
\item $R_P(i)$ has a higher priority than $R_P(i+l)$ for $l>0$
\item $R_C$ and $R_U$ have the lowest priority.
\end{itemize}

\section{Preliminary results}\label{sect:proof:deb}

\begin{lemma}\label{lem:noRootCreation}
  Let $\gamma^a\mapsto\gamma^b$ be a step.  If $p$ is a root in
  $\gamma^b$, then it also is in $\gamma^a$.
\end{lemma}
\begin{proof}
  We split the study into the following three cases.
  \begin{itemize}
  \item Suppose that $algoError(p)$ is true in $\gamma^b$.  Let $i$ be
    such that $\hei[b]{p} \geq i$, for each $q\in N(p)$,
    $\hei[b]{q}\geq i-1$, and $p^b.L[i]\neq \widehat\algo(p^b, i-1))$.

    The key element of this case is that, if $\hei[a]{q}\geq i$ and
    $\hei[b]{q}\geq i$, then $q^a.L[i]=q^b.L[i]$.
    \begin{itemize}
    \item If $\hei[a]{p}=i-1$, then $p$ applies the rule $R_U$ in
      $\gamma^a\mapsto\gamma^b$.  This implies that for each
      $q\in N[p]$, $\hei[a]{q}\geq i-1$ and
      $p^b.L[i]=\widehat\algo(p^a, i-1))$.  But since each $q\in N[p]$
      is such that $\hei[a]{q}\geq i-1$ and $\hei[b]{q}\geq i-1$, we
      have $q^a.L[i-1]=q^b.L[i-1]$ which contradicts that
      $p^b.L[i]\neq\widehat\algo(p^b, i-1))$.
    \item If $\hei[a]{p}\geq i$, and all $q\in N(p)$ are such that
      $\hei[a]{q}\geq i-1$, then again, $q^a.L|i-1]=q^b.L|i-1]$.  And
      thus $algoError(p)$ is true in $\gamma^a$.

    \item If $\hei[a]{p}\geq i$ and some $q\in N(p)$ is such that
      $\hei[a]{q}<i-1$, then $q$ cannot apply a rule $R_U$ in
      $\gamma^a\mapsto\gamma^b$, and thus $\hei[b]{q}<i-1$ which is a
      contradiction.
    \end{itemize}

  \item Suppose that $p.s^b=E$ and there exist no $q\in N(p)$ such
    that $q.s^b=E$ and $\hei[b]{q}<\hei[b]{p}$.

    If $p.s^a=E$ and no $q\in N(p)$ is such that $q.s^a=E$ and
    $\hei[a]{q}<\hei[a]{p}$, then $p$ is a root in $\gamma^a$.

    We claim that in all remaining cases, $p$ applies an error rule in
    $ \gamma^a\mapsto\gamma^b$.  Indeed
    \begin{itemize}
    \item if $p.s^a=E$ and there exists $q\in N(p)$ such that
      $q.s^a=E$ and $\hei[a]{q}<\hei[a]{p}$, then $q$ cannot apply a
      rule $R_U$ or $R_C$, and thus $q.s^b=E$.  Recall that in the
      current case, there exist no $u\in N(p)$ is such that $u.s^b=E$
      and $\hei[b]u<\hei[b]{p}$.  Thus $\hei[b]{q}\geq \hei[b]{p}$,
      which implies that $p$ must apply an error rule in
      $\gamma^a\mapsto\gamma^b$.
    \item if $p.s^a=C$ then $p$ must also apply an error rule in
      $\gamma^a\mapsto\gamma^b$.
    \end{itemize}

    Now 2 cases are possible.
    \begin{itemize}
    \item If $p$ applies a rule $R_R$, then $p$ is a root in
      $\gamma^a$.
    \item If $p$ applies a rule $R_P(i)$ in $\gamma^a\mapsto\gamma^b$,
      then there exists $q\in N(p)$ such that $\hei[a]{q}=i-1$ and
      $q.s^a=E$.  But since $q.s^a=E$, $q$ cannot apply the rule
      $R_U$, and because of $p$, $q$ cannot apply a rule $R_C$.  Thus
      $q.s^b=E$ and $\hei[b]{q}<\hei[b]{p}$, which contradicts the
      hypothesis.
    \end{itemize}

  \item Suppose that, $p.s^b=C$ and there exists $q\in N(p)$ such that
    $\hei[b]{q}\geq \hei[b]{p}+2$.

    Since $p$ does not apply an error rule in
    $\gamma^a\mapsto\gamma^b$, $\hei[a]{p}\leq \hei[b]{p}$.
    \begin{itemize}
    \item If $\hei[a]{q}=\hei[b]{q}-1$, $q$ applies the rule $R_U$ in
      $\gamma^a\mapsto\gamma^b$.  But this is impossible because
      $\hei[a]{q}> \hei[a]{p}$
    \item If $\hei[a]{q}\geq \hei[b]{q}$, then
      $\hei[a]{q}\geq \hei[a]{p}+2$.
      \begin{itemize}
      \item If $p.s^a=E$, then $p$ applies the rule $R_C$ in
        $\gamma^a\mapsto\gamma^b$, which is impossible because
        $\hei[a]{q}\geq \hei[a]{p}+2$.
      \item If $p.s^a=C$, then $p$ is a root in $\gamma$.
      \end{itemize}
    \end{itemize}
  \end{itemize}
  Thus, in all possible cases, $p$ is a root in $\gamma$.  \qed
\end{proof}

\begin{lemma}\label{lem:root_RC}
  Let $\gamma^a\mapsto\gamma^b$ be a step, and let $r$ be a root in
  $\gamma^a$ which applies the rule $R_C$ during
  $\gamma^a\mapsto\gamma^b$.  Then $\hei[a]{r}=0$ and $r$ is not a
  root in $\gamma^b$.
\end{lemma}
\begin{proof}
  Since $R_R$ has a higher priority than $R_C$, the guard of $R_R$ is
  false at $r$ in $\gamma^a$. So, as $r$ is a root in $\gamma^a$, we
  necessarily have $\hei[a]{r}=0$.

  Then, since $r$ applies $R_C$ during $\gamma^a\mapsto\gamma^b$, we
  have $r.s^b=C$.  Moreover, to allow $r$ to execute $R_C$, we should
  have every $q\in N(r)$ that satisfies $\hei[a]{q}\leq 1$. Now, as
  $\hei[a]{r}=0$, no $q\in N(r)$ with $\hei[a]{q}=1$ can apply $R_U$
  in $\gamma^a\mapsto\gamma^b$.  All this implies that
  $dependencyError(r)$ is false in $\gamma^b$.

  Finally, since $r$ applies the rule $R_C$, we have
  $\hei[b]{r}=\hei[a]{r}=0$, which implies that $algoError(r)$ is also
  false in $\gamma^b$.  Hence, $r$ is not a root in $\gamma^b$.  \qed
\end{proof}

A node $p$ is in \emph{error} in $\gamma$ if $p.s=E$ and it is a
\emph{root} if $root(p)$.  $R_R$ and each rule $R_P(i)$ are referred
to as \emph{error rules} in the following.

A path $P=p_0p_1\cdots p_l$ in $G$ is \emph{decreasing} in a
configuration $\gamma$ if for each $0\leq i<l$,
$\hei{p_i}>\hei{p_{i+1}}$.  A path $P$ is an \emph{$E$-path} if it is
decreasing, all its nodes are in error, and its last node is a root.

\begin{lemma}\label{lem:E-path}
  Let $\gamma$ be a configuration.  Any node $p$ in error is the first
  node of an $E$-path.
\end{lemma}
\begin{proof}
  We prove our lemma by induction on $\hei{p}$, if $\hei{p}=0$, then
  $p$ is a root and $P=p_0$ satisfies the required conditions.

  Suppose that $\hei{p}>0$.  If $p$ is a root, then $P=p_0$ satisfies
  the required conditions.  Otherwise, there exists $q\in N(v)$ such
  that $\hei{q}<\hei{p}$ and $q.s=E$.  By induction, there exists an
  $E$-path $P'$ starting a $q$.  We can add $p$ at the beginning of
  $P'$ to obtain a path $P$ which satisfies all required conditions.
  \qed
\end{proof}

Although this Lemma will have other important implications, it implies
that if there exists a node in error, then there exists a root in
error.

\section{Terminal configurations}

A configuration $\gamma$ is \emph{almost clean} if
\begin{itemize}
\item every root $r$ satisfies $\hei{r}=0$ and $r.s=E$, and
\item every two neighbor $p$ and $q$ satisfies
  $\bigl|\hei{p}-\hei{q}\bigr|\leq 1$.
\end{itemize}

\begin{lemma}\label{lem:almostClean}
  A configuration is almost clean if and only if no node can apply an
  error rule.
\end{lemma}
\begin{proof}
  Suppose that $\gamma$ is almost clean. Since every root $r$ is such
  that $\hei{r}=0$ and $r.s=E$, no node can apply the rule $R_R$, and
  since every neighbor $p$ and $q$ are such that
  $|\hei{p}-\hei{q}|\leq 1$, no node can apply a rule $R_P$.

  Conversely, suppose that $\gamma$ is not almost clean. If a root $r$
  is such that either $\hei{r}>0$ or $r.s=C$, then $r$ can apply the
  rule $R_R$. And if $\hei{q}\geq \hei{p}+2$, then either $p.s=C$ and
  $r$ is a root which can apply the rule $R_R$ or $p.s=E$ and $q$ can
  apply a rule $R_P$.  \qed
\end{proof}

\begin{lemma}\label{lem:almostCleanClosed}
  Let $\gamma^a\mapsto\gamma^b$ be a step. If $\gamma^a$ is almost
  clean, then so is $\gamma^b$.
\end{lemma}
\begin{proof}
  Assume, by the contradiction, that $\gamma^a$ is almost clean and
  $\gamma^b$ is not.

  At least one of these following two cases occur:
  \begin{itemize}
  \item Some node $p$ can apply the rule $R_P$ in $\gamma^b$. There
    are thus two neighbors $p$ and $q$ such that
    $\hei[b]{p}\geq \hei[b]{q}+2$.

    Since $\gamma^a$ is almost clean,
    $\bigl|\hei[a]{p}-\hei[a]{q}\bigr|\leq 1$ and no error rule is
    executed in the step $\gamma^a \mapsto \gamma^b$.  So,
    necessarily, $p$ executes $R_U$ in $\gamma^a \mapsto \gamma^b$,
    but not $q$. Thus, $\hei[a]{p} = \hei[b]{p} - 1$ and
    $\hei[a]{q} = \hei[b]{q}$. Moreover, since $p$ enabled for $R_U$
    in $\gamma^a$, we have $\hei[a]{p} \leq \hei[a]{q}$, which implies
    $\hei[b]{p} \leq \hei[b]{q}+1$, a contradiction.

  \item Some root $r$ can apply the rule $R_R$ in $\gamma^b$ (i.e.,
    $\hei[b]{r}\neq 0$ or $r.s^b=C$).

    First, by Lemma~\ref{lem:noRootCreation}, $r$ is a root in
    $\gamma^a$, and since $\gamma^a$ is almost clean, $\hei[a]{r}=0$
    and $r.s^a=E$. Thus, by Lemma~\ref{lem:root_RC}, either $r$
    applies no rule in $\gamma^a\mapsto\gamma^b$ which is a
    contradiction because $\hei[b]{r}=0$ and $r.s^b=E$, or $r$ applies
    the rule $R_C$ and $r$ is not a root in $\gamma^b$, which is also
    a contradiction.\qed
  \end{itemize}
\end{proof}

\begin{lemma}\label{lem:stableConfiguration}
  Let $\gamma$ be an almost clean configuration.  For any node $p$ and
  $i\leq \hei{p}$, $p.L[i]=\stateIA{p}{i}$.
\end{lemma}
\begin{proof}
  Since $\gamma$ is almost clean, for any neighboring nodes $p$ and
  $q$, $|\hei{p}-\hei{q}|\leq 1$.  By induction on $i\geq 0$, if $i=0$
  then we have set $p.L[i]:=p.init$, thus $p.L[i]=\stateIA{p}{i}$.
  Suppose that the induction hypothesis holds for some $i\geq 0$.  Let
  $p$ be such that $\hei{p}\geq i+1$.  Now by
  Lemma~\ref{lem:almostClean}, $p$ cannot apply an error rule,
  $algoError(p)$ is false and thus $p.L[i+1] =\widehat\algo(p, i)$.
  The induction hypothesis thus implies that
  $p.L[i+1]=\stateIA{p}{i+1}$. \qed
\end{proof}

A configuration in \emph{clean} if it contains no root.  The following
property gives an alternative definition of being clean, and as a
direct consequence, it implies that clean configurations are also
almost clean.

\begin{lemma}\label{lem:clean}
  A configuration is clean if and only if nodes can only apply the
  rule $R_U$.
\end{lemma}
\begin{proof}
  Suppose that $\gamma$ is clean. Since it contains no root, then no
  node can apply the rule $R_R$. Since there are no root, then, by
  Lemma~\ref{lem:E-path}, there are no node in error, and thus no node
  can apply a rule $R_P$ or the rule $R_C$.

  Conversely, suppose that nodes can only apply the rule $R_U$. Then
  by Lemma~\ref{lem:almostClean}, $\gamma$ is almost clean therefore
  $\gamma$ contains no correct root. To prove that $\gamma$ do not
  contain roots in error, it is enough to show that $\gamma$ contain
  no node in error. Suppose thus for a contradiction that $p$ is in
  error.  Since $\gamma$ is almost clean, every neighbor $q$ of $p$ is
  such that $|\hei{p}-\hei{q}|\leq 1$.  The only condition which
  prevents a node $p$ in error from applying the rule $R_C$ is thus if
  $p$ has a neighbor in error $q$ such that $\hei{q}=\hei{p}+1$. But
  then any node in error of maximum height can apply the rule $R_C$, a
  contradiction.\qed
\end{proof}

Lemma~\ref{lem:noRootCreation} implies that being clean is a stable
property.

A \emph{terminal configuration} is a configuration in which no node
can be activated. Clearly, by Lemma~\ref{lem:clean}, terminal
configurations are clean.
\begin{lemma}\label{lem:finalConfiguration}
  Suppose that $\gamma$ is a terminal configuration.  There exists $H$
  such that for any node $p$, $\hei{p}=H$ and $p.s=C$.  Moreover, for
  any $1\leq i\leq H$, $p.L[i]=\stateIA{p}{i}$.

  In greedy mode, every execution with $\durationIA=\infty$ is
  infinite, thus no terminal configuration exist. Otherwise,
  $H=\durationIA$.

  In lazy mode, if $\durationIA<\Time$, then
  $H=\durationIA$. Otherwise, $H\geq \Time$.
\end{lemma}
\begin{proof}
  Since $\gamma$ is also clean, the fact that for any
  $1\leq i\leq \hei{p}$, $p.L[i]=\stateIA{p}{i}$ then follows from
  Lemma~\ref{lem:stableConfiguration}.

  In greedy mode, if $\durationIA=\infty$, then every execution is
  infinite. Indeed, if $\gamma$ is not clean, then some node can apply
  an error rule or the rule $R_C$, and if $\gamma$ is clean, then any
  node of minimum height can apply the rule $R_U$. Otherwise, we claim
  that all nodes have the same height.  If not, there exists two
  neighboring nodes $p$ and $q$ such that $\hei{q}=\hei{p}+1$.  Among
  those pairs, choose one with $\hei{p}$ minimum.  Since, $p.s=C$,
  $\hei{q}>\hei{p}$ and for any neighbor $r\in N(p)$,
  $\hei{r}\geq \hei{p}$, then regardless of $\lazyI$, $p$ can apply
  the rule $R_U$.  Let $H$ be this common height.

  In greedy mode, if $\durationIA<\infty$, then $H=\durationIA$. And
  in lazy mode, if $H<\Time$, then there exists $p$ be such that
  $\stateIA{p}{H}\neq\stateIA{p}{H+1}$, which is only possible if
  $H=\durationIA$.  \qed
\end{proof}

\section{$D$-paths}

Recall that a path $P=p_0p_1\cdots p_l$ in $G$ is \emph{decreasing} in
a configuration $\gamma$ if for each $0\leq i<l$,
$\hei{p_i}>\hei{p_{i+1}}$ and that $P$ is an \emph{$E$-path} if it is
decreasing, all its nodes are in error, and its last node is a root.

We extend these definition in the following way.  A path $P$ is
\emph{gently decreasing} if, for each $0\leq i<l$,
$\hei{p_i}=\hei{p_{i+1}}+1$, and it is a \emph{$D$-path} if it is
decreasing and there exists $0\leq j\leq l$ such that
\begin{itemize}
\item $P_C=p_0\cdots p_{j-1}$ is a (possibly empty) gently decreasing
  path of nodes in $C$,
\item $P_E=p_j\cdots p_l$ is an $E$-path.
\end{itemize}
We call $P_C$ and $P_E$ the \emph{correct} and \emph{error} parts of
$P$.

\begin{lemma}\label{lem:D-path_height}
  Let $\gamma^a\mapsto\gamma^b$ be a step, and let $P$ be a $D$-path
  in $\gamma^a$.  For any $p\in P$, $\hei[b]{p}\leq \hei[a]{p}$.
  Moreover, if $p\in P$ is such that $p.s^b=C$, then we have equality.
\end{lemma}
\begin{proof}
  Let $p\in P$.  Recall that the height of a node $p$ only increases
  if $p$ applies the rule $R_U$.
  \begin{itemize}
  \item If $p$ is the last node of $P$, then in $\gamma^a$, $p$ is a
    root such that $p.s^a=E$.  Thus $p$ cannot apply the rule $R_U$ in
    $\gamma^a\mapsto\gamma^b$.
  \item If $p$ is not the last node of $P$, let $q$ be the next node
    after $p$ on $P$.  Since $P$ is decreasing in $\gamma^a$,
    $\hei[a]{q}<\hei[a]{p}$, and $p$ cannot apply the rule $R_U$ in
    $\gamma^a\mapsto\gamma^b$.
  \end{itemize}
  The first part of the lemma follows.  Now if $p.s^b=C$, then $p$
  does not apply an error rule in $\gamma^a\mapsto\gamma^b$, and thus
  $\hei[b]{p}\geq \hei[a]{p}$, which complete the proof.  \qed
\end{proof}

\begin{lemma}\label{lem:D-pathGlop}
  Let $\gamma^a\mapsto\gamma^b$ be a step, let $P=p_0\cdots p_l$ be a
  decreasing path in~$\gamma^a$ such that
  \begin{itemize}
  \item apart from $p_l$ which satisfies $p_l.s^a=E$ and $p_l.s^b=C$,
    all the nodes of $P$ are in $C$ in both $\gamma^a$ and $\gamma^b$;
  \item in $\gamma^a$, $p_0\cdots p_{l-1}$ is gently decreasing.
  \end{itemize}
  Then $P$ is gently decreasing in $\gamma^b$.
\end{lemma}
\begin{proof}
  The assumptions imply that $p_l$ applies the rule $R_C$ in the step
  $\gamma^a\mapsto\gamma^b$.  Thus $\hei[b]{p_l} = \hei[a]{p_l}$.
  Moreover, by Lemma~\ref{lem:D-path_height}, for any $0\leq i < l$,
  $\hei[b]{p_i} = \hei[a]{p_i}$.

  Since $p_l$ applies the rule $R_C$, we have
  $\hei[a]{p_{l-1}}\leq \hei[a]{p_l}+1$.  Together with the fact that
  $P$ is decreasing in $\gamma^a$, we obtain
  $\hei[a]{p_{l-1}}= \hei[a]{p_l}+1$.  The beginning of the path is
  gently decreasing by hypothesis, so $P$ is gently decreasing in
  $\gamma^a$.  Finally, since we have shown that the height of each of
  its nodes is the same in $\gamma^a$ and $\gamma^b$, the lemma
  follows.  \qed
\end{proof}

\begin{lemma}\label{lem:D-pathTrans}
  Let $\gamma^a\mapsto\gamma^b$ be a step, let $p$ be the first node
  of a $D$-path $P$ in~$\gamma^a$.  If at least one node of $P$ is in
  error in $\gamma^b$, then $p$ is the first node of a $D$-path in
  $\gamma^b$.
\end{lemma}
\begin{proof}
  Let $P=p_0\cdots p_l$ be a $D$-path in $\gamma^a$ and let $p=p_0$.
  Assume that $P$ contains at least one node in error in $\gamma^b$,
  and let $0\leq i\leq l$ be minimal such that $p_i.s^b=E$.

  Let $P'$ be the possibly empty path $p_0\cdots p_{i-1}$.  Since
  $p_i.s^b=E$, there exists a $E$-path $Q=p_iq_1\cdots q_h$ in
  $\gamma^b$, by Lemma~\ref{lem:E-path}.  We now claim that
  $P''=p_0\cdots p_iq_1\cdots q_h$ is a $D$-path in $\gamma^b$ whose
  first node is $p$.

  We first prove that $P''$ is decreasing.  Indeed, by
  Lemma~\ref{lem:D-path_height}, $\hei[b]{p_i}\leq \hei[a]{p_i}$ and,
  for $0\leq j < i$, $\hei[b]{p_j} = \hei[a]{p_j}$.  Since $P$ is
  decreasing in $\gamma^a$, the subpath $p_0\cdots p_i$ is also
  decreasing in~$\gamma^b$.  Now $p_iq_1\cdots q_h$ is an $E$-path and
  is thus also decreasing which implies that so is $P''$.

  To finish the proof, we must show that $P'$ is gently decreasing in
  $\gamma^b$.  First, if $P'$ is empty, we are done. Assume now that
  $P'$ is not empty.  In~$P$, only one node can apply the rule $R_C$
  in $\gamma^a\mapsto\gamma^b$: its first node in error in~$\gamma^a$.
  Now, since $p_i$ is the first node in error of $P$ in~$\gamma^b$, we
  have case: either $p_{i-1}.s^a=E$ and $p_{i-1}$ executes $R_C$ in
  $\gamma^a\mapsto\gamma^b$, or $p_{i-1}.s^a=C$ and does not execute
  $R_C$ in $\gamma^a\mapsto\gamma^b$.
  \begin{itemize}
  \item Suppose that $p_{i-1}.s^a=E$.  The path $P'$ is then gently
    decreasing in $\gamma^b$ by Lemma~\ref{lem:D-pathGlop}.
  \item Suppose that $p_{i-1}.s^a=C$.  Since $P$ is a $D$-path
    in~$\gamma^a$, $P'$ is gently decreasing in~$\gamma^a$.  Now, we
    have already shown that every correct node $q$ of $P'$ satisfies
    $\hei[a]{q}=\hei[b]{q}$.  The path $P'$ is thus gently decreasing
    in $\gamma^b$.  \qed
  \end{itemize}
\end{proof}

\begin{lemma}\label{lem:languageInSegment}
  Let $\gamma^a\mapsto\gamma^b$ be a step, let $p$ be the first node
  of a $D$-path, and let $r$ be its root in $\gamma^a$.  If $r$ is a
  root in $\gamma^b$, then $p$ is the first node of a $D$-path
  in~$\gamma^b$.
\end{lemma}
\begin{proof}
  If $r$ applies $R_C$ during $\gamma^a\mapsto\gamma^b$, then
  Lemma~\ref{lem:root_RC} implies that $r$ is not a root in
  $\gamma^b$, which is a contradiction.  Thus $r$ is in error in
  $\gamma^b$, and the Lemma follows from Lemma~\ref{lem:D-pathTrans}.
  \qed
\end{proof}

\begin{lemma}\label{lem:2theGround}
  Let $\gamma^a\mapsto\gamma^b$ be a step, and let $p$ be the first
  node of a $D$-path in~$\gamma^a$.  If no $D$-path in $\gamma^b$
  contains $p$, then $\hei[b]{p}\leq n$.
\end{lemma}
\begin{proof}
  Let $p$ be the first node of a $D$-path $P$, and let $r$ be the root
  of $P$ in~$\gamma^a$.

  We claim that, in $\gamma^b$, $P$ contains no node in error.
  Indeed, otherwise Lemma~\ref{lem:D-pathTrans} implies that $p$ is
  the first node of a $D$-path in $\gamma^b$, which is a
  contradiction.

  Since, in a $D$-path, at most one node can apply the rule $R_C$
  during a step then, in $\gamma^a$, all the nodes of $P$ but $r$ have
  status $C$.  We can thus apply Lemma~\ref{lem:D-pathGlop}, and
  obtain that $P$ is gently decreasing in $\gamma^b$, and thus
  $\hei[b]{p}=lenght(P)+\hei[b]{r}$.  And since no node can appear
  twice in $P$, $\hei[b]{p} \leq \hei[b]{r}+n$.

  Now since $r.s^b=C$, $r$ applies the rule $R_C$ in
  $\gamma^a\mapsto\gamma^b$.  But then Lemma~\ref{lem:root_RC} implies
  that $\hei[a]{r}=0$, and thus $\hei[b]{r}=0$.  The lemma follows.
  \qed
\end{proof}

\section{Moves complexity}

In this section, we analyze the move complexity of our algorithm. To
that goal, we fix an execution $e=\gamma^0\gamma^1\cdots$ and study
the rules a given node applies in.  Since, these rules do not appear
explicitly in an execution, we propose to use a proxy for them.

A pair $(p, i)$ is a \emph{move} if $p$ applies a rule in
$\gamma^i\mapsto\gamma^{i+1}$.  This move is a \emph{$U$-move} if the
rule is $R_U$, a \emph{$C$-move} if the rule is $R_C$, a
\emph{$R$-move} if the rule is $R_R$, and a \emph{$P(i)$-move} if the
rule is $R_P(i)$.  Since a node $p$ applies at most one rule in a
given step, the number of steps in which a given node applies a rule
is the number of its moves, and the number of steps is bounded by the
total number of moves.

\subsection{$R$-moves}

\begin{lemma}\label{lem:nbRmoves}
  During an execution, there are at most $n$ $R$-moves.
\end{lemma}
\begin{proof}
  Let $p$ be a node. We claim that $p$ can apply the rule $R_R$ at
  most once.  We have three cases.
  \begin{itemize}
  \item If $p$ executes no $R$-move, it executes at most one $R$-move.

  \item If $p$ executes a $R$-move and no move after the first
    $R$-move, then $p$ executes only one $R$-move.

  \item Otherwise, let $(p,i)$ be the first $R$-move, and let $(p,j)$
    be the first move which follows.

    Since $(p,i)$ is a $R$-move, $p.s^{i+1}=E$ and $\hei[i+1]{p}=0$,
    and since $p$ applies no rule between $\gamma^{i+1}$ and
    $\gamma^j$, $p.s^j=E$ and $\hei[j]{p}=0$.  Consequently, $(p,j)$
    is necessarily a $C$-move, and thus, by Lemma~\ref{lem:root_RC},
    $p$ is not a root in
    $\gamma^{j+1}$. Lemma~\ref{lem:noRootCreation} then implies that
    $p$ is a root in no $\gamma^h$ for $h>j$, and thus no $(p, h)$ is
    an $R$-move for $h>j$.  Hence, $p$ executes only one $R$-move in
    this case.
  \end{itemize}
  Hence, in all cases, $p$ makes at most one $R$-move. The Lemma
  follows.  \qed
\end{proof}

\subsection{$U$-move}

Let $S_i$ be the set of roots in $\gamma^i$.
Lemma~\ref{lem:noRootCreation} states that for each $i>0$,
$S_{i}\subseteq S_{i-1}$.  Since $\gamma^0$ contains at most $n$
roots, there are $l\leq n$ steps $\gamma^{i-1}\mapsto\gamma^{i}$ for
which $S_{i}\subset S_{i-1}$. Let $r_1$, $r_2, \ldots, r_l$ be the
sequence of increasing indices such that for all
$i \in [1..l], S_{r_i}\subset S_{r_i-1}$. This sequence gives the
following \emph{decomposition} of $e$ into segments.

\begin{itemize}
\item The \emph{first segment} is the sequence
  $\gamma^0\cdots\gamma^{r_1}$.
\item For $1<i\leq l$ the \emph{$i$-th segment} is the sequence
  $\gamma^{r_{i-1}}\cdots\gamma^{r_i}$.
\item The \emph{last segment} is the sequence $\gamma^{r_l}\cdots$.
\end{itemize}

At this point in the proof, it is not obvious wether $\gamma^{r_i}$
should be clean or not. But what is clear is that because of
Lemma~\ref{lem:noRootCreation}, if some $\gamma^{r_i}$ is clean, then
so are all the other configuration in the segment Moreover, there is
at most one clean segment, and if it exists, it must be the last
segment.

The key Lemma to bound the number of $U$-moves is the following Lemma.

\begin{lemma}\label{lem:boundedIncreaseInSegment}
  In a segment, the number of times that a node applies the rule $R_U$
  is either infinite of is equal to the maximum of
  $\hei[j]{p}-\hei[i]{p}$ with $i<j$ in the segment.
\end{lemma}
\begin{proof}
  Since $\hei{p}$ increases by one each time that $p$ applies the rule
  $R_U$, if $p$ applies this rule a finite number of times, then the
  maximum of $\hei[j]{p}-\hei[i]{p}$ with $i<j$ in the segment is
  bounded.

  Now if $p$ applies an error rule, then it becomes in error.  By
  Lemma~\ref{lem:E-path}, $p$ is then the first node of an $E$-path,
  and thus of a $D$-path, by definition.
  Lemma~\ref{lem:languageInSegment} then implies that $p$ remains in a
  $D$-path until the end of the segment.
  Lemma~\ref{lem:D-path_height} finally implies that $p$ does not
  apply the rule $R_U$ until the end of the segment.

  Since $p.L$ decreases only when $p$ executes an error rule, all this
  implies that the number of times that $p$ applies the rule $R_U$ is
  equal to the maximum of $\hei[j]{p}-\hei[i]{p}$ with $i<j$ in the
  segment.  \qed
\end{proof}

To bound the number of $U$-moves, it is thus enough to bound
$\hei[j]{p}-\hei[i]{p}$ in a segment.

\begin{lemma}\label{lem:2D+1}
  If $i<j$ and $p$ satisfy $\hei[j]{p}> \hei[i]{p}+2D$, then for any
  $q$, there exists $i\leq h<j$ such that $\hei[h]{q}=\hei[i]{p}+D$
  and $(q, h)$ is a $U$-move.
\end{lemma}
\begin{proof}
  We prove by induction on $d(q, p)$ that there exist
  $i\leq i'<j'\leq j$ such that $\hei[i']{q}\leq \hei[i]{p}+d(q, p)$
  and $\hei[j]{p}-d(q, p)\leq \hei[j']{q}$.

  \begin{itemize}
  \item If $d(q, p)=0$, then $q=p$ and $i'=i$ and $j'=j$ do the trick.

  \item If $d(q, p)>0$ then let $q'\in N(q)$ be such that
    $d(q',p)=d(q, p)-1$.  By induction, there exists
    $i\leq i_1<j_1\leq j$ such that
    $\hei[i_1]{q'}\leq \hei[i]{p}+d(q, p)-1$ and
    $\hei[j]{p}-d(q, p)+1\leq \hei[j_1]{q'}$.

    Now
    $\hei[j_1]{q'}-\hei[i_1]{q'}\geq \hei[j]{p}-\hei[i]{p}-2(d(q,
    p)-1)>2D-2(d(q, p)-1)>2$.  There thus exists $i_1\leq i'<j_1$ such
    that $\hei[i']{q'}=\hei[i_1]{q'}$ and $(q', i')$ is a $U$-move.
    In $\gamma^{i'}$, every neighbor of $q'$ and thus $q$ satisfy that
    $\hei[i']{q}\leq \hei[i']{q'}+1=\hei[i_1]{q'}+1\leq
    \hei[i]{p}+d(q, p)$.

    Now since $\hei[i'+1]{q'}+2\leq\hei[j_1]{q'}$, there exists
    $i'< j'<j_1$ such that $(q', j')$ is a $U$-move and
    $\hei[j'+1]{q'}=\hei[j_1]{q'}$.  In $\gamma^{j'}$, every neighbor
    of $q'$ and thus $q$ satisfy that
    $\hei[j']{q}\geq \hei[j']{q'}=\hei[j_1]{q'}-1\geq \hei[j]{p}-d(q,
    p)$ which finishes the proof of our induction.
  \end{itemize}

  Let $q$ be any node.  Let $i\leq i'<j'\leq j$ such that
  $\hei[i']{q}\leq \hei[i]{p}+d(p, q)\leq \hei[i]{p}+D$ and
  $\hei[j']{q}\geq \hei[j]{p}-d(q, p)\geq \hei[j]{p}-D> \hei[i]{p}+D$.
  There exists $i'\leq h<j'$ such that $\hei[h]{q}=\hei[i]{p}+D$ and
  $(q, h)$ is a $U$-move.  \qed
\end{proof}

\begin{lemma}\label{lem:boundedIncreaseWithRoot}
  If $\gamma^h$ in not clean, then for any node $p$ and any
  $i<j\leq h$, $\hei[j]{p}-\hei[i]{p}\leq 2D$.
\end{lemma}
\begin{proof}
  Let $r$ be a root in $\gamma^h$ and let $i<j\leq h$.  Since no root
  can apply the rule $R_U$, the lemma follows from
  Lemma~\ref{lem:2D+1}.  \qed
\end{proof}

\begin{lemma}\label{lem:nbUmovesWithRoots}
  A node $p$ executes at most $\min(n\durationIA, 2nD)$ $U$-moves
  before the first clean configuration.
\end{lemma}
\begin{proof}
  By Lemma~\ref{lem:noRootCreation}, there are at most $n$ non clean
  segments.  By Lemma~\ref{lem:boundedIncreaseInSegment}, in each of
  these segments, the number of times that $p$ applies the rule $R_U$
  is equal to the maximum of $\hei[j]{p}-\hei[i]{p}$ with $i<j$ in the
  segment.

  We can bound $\hei[j]{p}-\hei[i]{p}$ by $\durationIA$, which gives a
  total bound of $n\durationIA$. But because of the roots,
  Lemma~\ref{lem:boundedIncreaseWithRoot} allows us to bound
  $\hei[j]{p}-\hei[i]{p}$ by $2D$, giving a total bound of $2nD$. \qed
\end{proof}

By Lemma~\ref{lem:noRootCreation}, there is at most 1 clean segment.
Of course, we can also bound the number of $U$-moves of a given node
in this segment (if it exists) by $\durationIA$, but when in lazy
mode, we can do better.

\begin{lemma}\label{lem:boundedIncreaseWithoutRootLazy}
  Let $\gamma^h \cdots$ be the clean segment (if it exists).  If the
  algorithm runs in lazy mode, then for any $i$, $j$ such that
  $h\leq i\leq j$ and any node $p$,
  $\hei[j]{p}-\hei[i]{p}\leq \max(\Time, D)$.
\end{lemma}
\begin{proof}
  First note that Lemma~\ref{lem:noRootCreation} implies that all
  configuration are clean.

  Let $H=\max_p(\hei[h]{p})$, and let $\bar H=\max(H, \Time)$.  We
  claim that for any $j\geq i$ and any node $p$,
  $\hei[j]{p}\leq \bar H$.  Indeed, by induction on $j\geq i$, this is
  true for $i=j$.  Let us now suppose that the property is true for
  some $j\geq i$, and let $p$ be any node.
  \begin{itemize}
  \item If $\hei[j]{p}<\bar H$, then $\hei[j+1]{p}\leq \bar H$.
  \item If $\hei[j]{p}=\bar H$ and at least one $q\in N(p)$ is such
    that $\hei[j]{q}<\bar H$ then $p$ cannot apply the rule $R_U$ in
    $\gamma^j\mapsto\gamma^{j+1}$.  Thus $\hei[j+1]{p}\leq H$.
  \item If $\hei[j]{p}=\bar H$ and all $q\in N(p)$ are such that
    $\hei[j]{q}=\bar H$, then Lemma~\ref{lem:stableConfiguration} and
    the fact that $H\geq \Time$ imply that
    $p.L^j[\bar H]=\widehat\algo(p^j, \bar H)$.  Hence, $p$ cannot
    apply the rule $R_U$, and $\hei[j+1]{p}\leq \bar H$.
  \end{itemize}
  Our claim is thus valid. Obviously, if $\bar H=\Time$, then the
  Lemma is true.  Let us thus suppose that $\bar H=H>\Time$.  In this
  case, it is enough to prove that for any $p$, $\hei[i]{p}\geq H-D$.

  Indeed, otherwise, if $\hei[i]{p}< H-D$, a shortest path between any
  $r$ such that $\hei[i]{r}=H$, and $p$ is not gently
  decreasing. There thus exist neighboring nodes $q$ and $q'$ such
  that $\hei[i]{q'}\geq\hei[i]{q}+2$. If $q.s^i=C$, then $q$ is a
  root, which is a contradiction. And if $q.s^i=E$, then by
  lemma~\ref{lem:E-path}, $q$ is the first node of an $E$-path, which
  implies that $\gamma^i$ contains a root, also a contradiction.  \qed
\end{proof}

\begin{lemma}\label{lem:nbUmoves}
  During an execution, the number of $U$-moves that our algorithm
  executes is infinite in greedy mode when $\durationIA=\infty$ and at
  most
  \begin{itemize}
  \item $O(n^2\min(\durationIA, D)+n\durationIA))$ in greedy mode with
    $\durationIA<\infty$,
  \item $O(n^2\min(\durationIA, D)+n\Time))$ in lazy mode.
  \end{itemize}
\end{lemma}
\begin{proof}
  The fact that this number is infinite when $\durationIA=\infty$ in
  greedy mode follows from Lemma~\ref{lem:finalConfiguration}.
  Otherwise, we must add to the bound of
  Lemma~\ref{lem:nbUmovesWithRoots}, the number of $U$-moves in the
  clean segment (if it exists).

  In greedy mode, for a given node, we bound $\hei[j]{p}-\hei[i]{p}$
  by $\durationIA$ which in now assumed to be finite. In lazy mode, we
  using Lemma~\ref{lem:boundedIncreaseWithoutRootLazy} while
  remembering that $\hei[j]{p}-\hei[i]{p}$ is always at most
  $\durationIA$.  \qed
\end{proof}

\subsection{$P$-move}

To count the number of $P$-moves of a given node $p$, we need several
definitions.

We say that a $P$-move $(p, t)$ \emph{causes} another $P$-move
$(p', t')$ if
\begin{itemize}
\item $p'\in N(p)$, $t'>t$,
\item for some $l$, $(p', t')$ is a $P(l)$-move and $(p, t)$ is a
  $P(l-1)$-move, and
\item for any $t<k<t'$, $(p, k)$ is not a move.
\end{itemize}

If a node $p$ is in error in some configuration $\gamma^i$, this often
happens because of some previous $P$-move $(p, t)$.  Moreover, what
allowed $(p, t)$ is some $q\in N(p)$ which is in error in
$\gamma^{t-1}$.  Finally, the reason why $q$ is in error in
$\gamma^{t-1}$ is because of some previous move and so on.  This
motivates the following definition: a \emph{causality chain} is a
sequence $C=(p_0, t_0)(p_1, t_1)\cdots(p_l, t_l)$ such that
\begin{itemize}
\item for each $0\leq i<l$, $(p_i, t_i)$ causes $(p_{i+1}, t_{i+1})$;
\item no $(p, t)$ causes $(p_0, t_0)$.
\end{itemize}

By construction, any $P$-move is the last element of a causality chain
but the causality chain may not be unique.

We classify the $P$-move of $p$ in 2 types.
\begin{itemize}
\item $(p, i)$ is of Type 1 if there exists a $P$-move $(p, j)$ with
  $j>i$ such that $\hei[i+1]{p}=\hei[j+1]{p}$.
\item $(p, i)$ is of Type 2 otherwise.
\end{itemize}

Our goal is to separately bound the number of $P$-moves of each type
that a node can execute.

\begin{lemma}\label{lem:nbType1}
  There are at most as many $P$-moves of type 1 as there are $U$-moves
  before the first clean configuration.
\end{lemma}
\begin{proof}
  Suppose that $(p, i)$ and $(p, j)$ are both $P(l)$-moves with $i<j$,
  let thus $l:=\hei[i+1]{p}$ ($=\hei[j+1]{p})$.  For $(p, j)$ to be
  possible, $\hei{p}$ has to go from $l$ in $\gamma^{i+1}$ to being
  strictly greater than $l$ in $\gamma^j$.  This implies that there
  exists $i<k<j$ such that $(p, k)$ is a $U$-move with $\hei[k]{p}=l$.

  Thus, if we associate to each $(p, i)$ of Type 1 the $U$-move
  $(p, j)$ such that $\hei[i+1]{p}=\hei[j]{p}$ with $j>i$ minimum,
  then no 2 distinct $P$-moves correspond to the same $U$-move.
  Moreover, since this $U$-move comes before a $P$-move, by
  Lemma~\ref{lem:almostClean}, this means that the corresponding
  configuration is not almost clean, and thus not clean.  All this
  implies that $p$ executes at most as many $P$-moves of Type 1 as
  $U$-moves before the first clean configuration.  \qed
\end{proof}

Remark that, by definition, two $P$-moves $(p, i)$ and $(p, j)$ of
Type 2 are such that $\hei[i+1]{p}\neq \hei[j+1]{p}$.  To bound the
number of $P$-moves $(p, i)$ of Type 2, we thus count the number of
values that $\hei[i+1]{p}$ can take.

The following Lemma is a direct consequence of this remark.
\begin{lemma}\label{lem:nbType2Bound}
  There are at most $n\durationIA$ $P$-moves of type 2.
\end{lemma}

Since we also want a bound which does not rely on the value of
$\durationIA$ (which can be $\infty$), we need to be more precise.  To
do so, we subdivide Type 2 $P$-moves in
\begin{itemize}
\item Type 2a.  if at least one causality chain
  $C=(p_0, t_0)\cdots(p_l, t_l)$ ending in $(p, i)$ does not contain a
  repeated node.  More formally, for any $0\leq i<j\leq l$,
  $p_i\neq p_j$.
\item Type 2b.  otherwise.
\end{itemize}

\begin{lemma}\label{lem:nbType2a}
  A node $p$ executes at most $n(n+1)$ Type 2a $P$-moves.
\end{lemma}
\begin{proof}
  Let $(p, i)$ be a $P$-move of Type 2a, and let
  $C=(p_0, t_0)\cdots(p_l, t_l)$ be a corresponding causality chain.
  We have
  \begin{itemize}
  \item $(p, i)=(p_l, t_l)$
  \item for any $0\leq i<j\leq l$, $p_i\neq p_j$.
  \end{itemize}
  Clearly, $l<n$ and $\hei[t_l+1]{p_l}=l+\hei[t_0+1]{p_0}$.  Let
  $r\in N(p_0)$ be such that $r.s^{t_0}=E$ and
  $\hei[t_0+1]{p_0}=\hei[t_0]{r}+1$.  Since no $P$-move causes
  $(p_0, t_0)$, two cases arise:
  \begin{itemize}
  \item the last move of $r$ before $t_0$ is an $R$-move in which case
    $\hei[t_0]{r}=0$,
  \item $r$ applies no rule before $t_0$ in which case
    $\hei[t_0]{r}=\hei[0]{r}$.
  \end{itemize}
  Thus $\hei[t_0+1]{p_0}$ can have at most $n+1$ distinct values.  The
  lemma now follows from the fact that $l$ can also take at most $n$
  distinct values.  \qed
\end{proof}

\begin{lemma}\label{lem:nbType2b}
  A node $p$ executes at most $2(n+D)$ Type 2b $P$-moves.
\end{lemma}
\begin{proof}
  Let $(p, h)$ be a $P$-move of Type 2b, and let
  $C=(p_0, t_0)\cdots(p_l, t_l)$ be a causality chain such that
  $(p, h)=(p_l, t_l)$.

  By definition, there exists $0\leq i<j\leq l$ such that $p_i=p_j$.
  Choose such a $i_0=i$ and $j_0=j$ with $j_0$ maximum.  We thus have
  that for any $j_0\leq i<j\leq l$, $p_i\neq p_j$ and thus $l-j_0<n$.
  Let $q=p_{j_0}=p_{i_0}$.

  Now $\hei[l+1]{p_l}=\hei[j_0+1]{q}+(l-j_0)$.  To prove the lemma, it
  is thus enough to show that $\hei[j_0+1]{q}\leq n+2D$.

  We have that $q.s^{i_0+1}=E$, thus, by Lemma~\ref{lem:E-path}, $q$
  is the first node of an $E$-path, and thus of a $D$-path in
  $\gamma^{i_0+1}$.

  Since $\hei[j_0+1]{q}>\hei[i_0+1]{q}$, $q$ applies a $U$-move
  $(q, k)$ for $i_0<k< j_0$.  By Lemma~\ref{lem:D-path_height}, $q$
  belongs to no $D$-path in $\gamma^k$.  There thus exists
  $i_0\leq k'<k$ such that $q$ belongs to a $D$-path in $\gamma^{k'}$
  and to no $D$-path in $\gamma^{k'+1}$.  By
  Lemma~\ref{lem:2theGround}, $\hei[k'+1]{q}\leq n$.

  Since $q$ is in error in $\gamma^{j_0}$, by Lemma~\ref{lem:E-path},
  $q$ belongs to an $E$-path in $j_0$.  There thus exists a root $r$
  in $\gamma^{j_0}$.  By Lemma~\ref{lem:boundedIncreaseWithRoot},
  $\hei[j_0+1]{q}-\hei[k'+1]{q}\leq 2D$, and thus
  $\hei[j_0+1]{q}\leq n+2D$. The lemma follows.  \qed
\end{proof}

Lemmas~\ref{lem:nbType1},~\ref{lem:nbType2a}, and~\ref{lem:nbType2b}
directly imply the following Lemma.
\begin{lemma}\label{lem:nbPmoves}
  There are at most $O(n^2\min(\durationIA, n))$ $P$-moves during an
  execution before the first clean configuration.
\end{lemma}

\subsection{$C$-moves}

\begin{lemma}\label{lem:nbCmoves}
  During an execution, the number of $C$-moves before the first clean
  configuration is at most the number of $P$-moves plus $n$.
\end{lemma}
\begin{proof}
  Between 2 $C$-moves, a node $p$ must execute an error move.

  But since, after a $C$-move, $p$ is can no longer be a root (by
  Lemma~\ref{lem:root_RC} and \ref{lem:noRootCreation}), $p$ cannot
  execute a $C$-move before an $R$-move. Thus $p$ can execute at most
  one more $C$-move than its number of $P$-moves.  \qed
\end{proof}

\subsection{The move complexity theorem}

The following theorem is a direct corollary of Lemmas
\ref{lem:finalConfiguration}, \ref{lem:nbRmoves}, \ref{lem:nbUmoves},
\ref{lem:nbPmoves} and \ref{lem:nbCmoves}.

\begin{theorem}\label{thm:stepComplexity}
  In any execution, our algorithm always reaches a clean configuration
  in at most $O(n^2\min(\durationIA, n))$moves.

  It does not finish if $\durationIA=\infty$ in greedy mode and
  otherwise, it executes at most:
  \begin{itemize}
  \item $O(n^2\min(\durationIA, n)+n\durationIA))$ moves in greedy
    mode.
  \item $O(n^2\min(\durationIA, n)+n\Time))$ moves in lazy mode.
  \end{itemize}
\end{theorem}

\section{Round complexity proof}\label{sect:round}

Through out this section, we consider an arbitrary execution
$e = \gamma^0 \cdots$.  Let
$\gamma^0\cdots\gamma^{h_1}\cdots\gamma^{h_2}\cdots\gamma^f$ be a
decomposition of $e$ into non-empty rounds (n.b., $e$ is finite, by
Theorem~\ref{thm:stepComplexity}). We also let
$\gamma^{h_0} = \gamma^0$.

\subsection{The ``error broadcast phase''}

\begin{lemma}\label{lem:errorFirstRound}
  Let $r$ be a root in $\gamma^h$ for $h\geq h_1$.  Then
  $\hei[h]{r}=0$ and $r.s=E$.
\end{lemma}
\begin{proof}
  We claim that there exists a configuration $\gamma^i$ before the end
  of the first round (i.e., $i\leq h_1$) such that $\hei[i]{r}=0$ and
  $r.s^i=E$.

  If it is not the case in $\gamma^0$, then by
  Lemma~\ref{lem:noRootCreation}, $r$ is a root which can apply the
  rule $R_R$ in $\gamma^0\mapsto\gamma^1$.  Since $r$ cannot be
  disabled, the claimed $i$ exists.

  Now for the state of $r$ to change, it must apply the rule $R_C$.
  But as soon as $r$ does so, by Lemma~\ref{lem:root_RC}, it no longer
  is a root, and, by Lemma~\ref{lem:noRootCreation}, will never be a
  root again.  Since $r$ is a root in $\gamma^h$, the state of $r$
  does not change between $\gamma^i$ and $\gamma^h$.  \qed
\end{proof}

\begin{lemma}\label{lem:penteDouce}
  For any root $r$ in $\gamma^h$ with $h\geq h_{d+1}$, for any node
  $p$ such that $d(p, r)\leq d$, we have $\hei[h]{p}\leq d(p, r)$.
\end{lemma}
\begin{proof}
  If $\gamma^h$ contains no root, then the lemma is true.  Otherwise,
  let $r$ be a root in $\gamma^h$. We prove the lemma by induction on
  $i=d(p,r)$ in this case.

  If $i=0$, then $p=r$.  The base case directly follows from
  Lemma~\ref{lem:errorFirstRound}.

  Suppose now that $i\geq 1$, and let $p$ be a node such that
  $d(r, p)=i$.  Let $q\in N(p)$ such that $d(r, q)=i-1$.

  We first show that there exists $h_i\leq j\leq h_{i+1}$ such that
  $\hei[j]{p}\leq i$.  If $\hei[h_i]{p}\leq i$, then $j=h_i$ and we
  are done.  Otherwise, $\hei[h_i]{p} > i$ and assume, by
  contradiction that $\hei[j]{p} > i$ for any
  $h_i \leq j \leq h_{i+1}$.  Now, by induction hypothesis,
  $\hei[j]{q}\leq i-1$ for any $j \geq h_i$.  So,
  $\hei[j]{p}\geq \hei[j]{q}+2$ for any $h_i \leq j \leq h_{i+1}$.
  Lemma~\ref{lem:noRootCreation} and~\ref{lem:errorFirstRound} implies
  that $p$ is not a root, and cannot apply the rule $R_R$ in any
  $\gamma^j\mapsto\gamma^{j+1}$ for any $h_i \leq j \leq h_{i+1}$.
  Moreover, the node $q$ is in error in $\gamma^{j}$ (for any
  $h_i \leq j \leq h_{i+1}$) as otherwise it would be a root not in
  error, contradicting Lemma~\ref{lem:errorFirstRound}.  Thus, $p$ is
  enabled for rule $R_P$ in $\gamma^{j}$ for any
  $h_i \leq j \leq h_{i+1}$(recall that we already have proven that
  $p$ cannot apply the higher-priority rule $R_R$). By definition of a
  round, $p$ executes $R_P$ during the $i+1$-th round, which leads to
  a contradiction. Hence, there exists $h_i\leq j\leq h_{i+1}$ such
  that $\hei[j]{p}\leq i$.

  To finish, notice that, for any $k \geq j \geq h_i$,
  $\hei[k]{q}\leq i-1$ (by induction hypothesis), which prevents $p$
  from applying the rule $R_U$ so that $\hei{p}>i$, and we are done.
  \qed
\end{proof}

\begin{lemma}\label{lem:roundD}
  For any $h\geq h_{\min(\durationIA, D)+1}$, $\gamma^h$ is almost
  clean.
\end{lemma}
\begin{proof}
  Assume, by the contradiction, that $\gamma^h$ is not almost clean if
  $h\geq h_{D+1}$ or $h\geq h_{\durationIA+1}$.

  In either case, $h \geq h_1$.  So, the first part of the almost
  clean definition holds in the two considered cases, by
  Lemma~\ref{lem:errorFirstRound}. Thus, there are two neighbors $p$
  and $q$ such that $\hei[h]{p}\geq \hei[h]{q}+2$.  The node $q$ is in
  error in $\gamma^h$ as otherwise it would be a root not in error,
  contradicting Lemma~\ref{lem:errorFirstRound}.  By
  Lemma~\ref{lem:E-path}, $q$ is the first node of an $E$-path $P$. By
  definition, $P$ leads to some root $r$ and
  $\hei[h]{q}\geq l+\hei[h]{r}$, where $l$ is the length of $P$.
  Since $\hei[h]{r}\geq 0$ and $l \geq d(q,r)$ (by definition), we
  have $\hei[h]{q}\geq d(q,r)$.  Moreover, by
  Lemma~\ref{lem:noRootCreation}, $r$ is already a root in $\gamma^0$.

  We now consider each of the two cases:
  \begin{itemize}
  \item If $h\geq h_{D+1}$, then, by definition, $d(q, r) \leq D$ and
    $d(p, r) \leq D$.

  \item If $h\geq h_{\durationIA+1}$, then, by definition and
    hypothesis, $\hei[h]{q}+2 \leq \hei[h]{p} \leq \durationIA$. So,
    $d(q, r)+2 \leq \durationIA$ and
    $d(p, r)+1 \leq d(q, r)+2 \leq \durationIA$.

  \end{itemize}
  Hence, in each case, we can apply Lemma~\ref{lem:penteDouce} with
  $h_{d+1} = h_{D+1}$ and $h_{d+1} = h_{\durationIA+1}$, respectively.
  Thus $\hei[h]{q} \leq d(q, r)$ and
  $\hei[h]{p} \leq d(p, r) < d(q,r)+2$. But, this implies that
  $\hei[h]{q} = d(q, r)$ and $\hei[h]{p} < \hei[h]{q} +2$ with is a
  contradiction.\qed
\end{proof}

\subsection{The ``error cleaning phase''}

Round complexity proofs are often tedious because if a node $p$ can
apply a rule $X$ at the beginning of a round, it can apply this rule
before the end of the round but it may end up applying another rule
$Y$ or be deactivated.  The following Lemma proves that, as soon as
the system has converged to almost clean configurations, only the
first case happens.  To avoid a lot of technicalities, we will thus
use it implicitly.

\begin{lemma}\label{lem:noDeactivation}
  Assume that for any $h\geq 0$, $\gamma^h$ is almost clean.  If a
  node $p$ can apply a rule $X\in \{R_C, R_U\}$ at the beginning of a
  round, then at the end of this round, $p$ will have executed $X$.
\end{lemma}
\begin{proof}
  Since all $\gamma^h$ are almost clean, no node can apply an error
  rule.  So, $X=R_C$ if and only if $p.s=E$, and $X=R_U$ if and only
  $p.s=C$.  We thus only have to prove that $p$ cannot be deactivated.

  By contradiction, let $\gamma^i$ be the first configuration such
  that $p$ has been deactivated.  We thus have that $p$ can apply the
  rule $X$ in $\gamma^{i-1}$, $p.s^{i-1}=p.s^i$ and $p.L^{i-1}=p.L^i$,
  and therefore $\hei[i-1]{p}=\hei[i]{p}$.

  If $X=R_C$, then in $\gamma^i$, there exists $q\in N(p)$ such that
  $\hei[i]{q}=\hei[i]{p}+1$ and $q.s^i=E$ or there exists $q\in N(p)$
  such that $\hei[i]{q}\geq \hei[i]{p}+2$.
  \begin{itemize}
  \item In the first case, since $q.s^i=E$, then either $q$ applies an
    error rule (which is imposible) or $q$ applies no rule in
    $\gamma^{i-1}\mapsto\gamma^i$.

    Thus $p$ is already deactivated in $\gamma^{i-1}$, a
    contradiction.

  \item In the second case, $\bigl|\hei[i]{p}-\hei[i]{q}\bigr|\geq 2$,
    and thus $\gamma^i$ is not almost clean, a contradiction.
  \end{itemize}

  If $X=R_U$, then in $\gamma^i$, there exists $q\in N(p)$ such that
  $\hei[i]{q}\notin\{\hei[i]{p}, \hei[i]{p}+1\}$ or both for all
  $q\in N(p),\; \hei[i]{q}=\hei[i]{p}$, and
  $p.L^i[\hei[i]{p}]= \widehat\algo(p^i, \hei[i]{p})$.

  \begin{itemize}
  \item In the first case, since $\gamma^i$ is almost clean,
    $\hei[i]{q}$ cannot be greater than $\hei[i]{p}+1$.  Thus
    $\hei[i]{q}<\hei[i]{p}$.  But since $\hei[i-1]{q}\leq\hei[i]{q}$
    (recall that no error rule can be applied) and
    $\hei[i]{p}=\hei[i-1]{p}$, $p$ cannot apply the rule $R_U$ in
    $\gamma^{i-1}$, a contradiction.
  \item In the second case, let $l=\hei[i]{p}=\hei[i-1]{p}$.

    If there exists $q\in N(p)$ such that $\hei[i-1]{q}<l$, then $p$
    cannot apply $R_U$ in $\gamma^{i-1}$, which is not the case.

    We therefore have that for all $q\in N(p)$, $\hei[i-1]{q}=l$
    (recall that no error rule can be applied).  We thus have that
    \begin{eqnarray*}
      p.L^{i-1}[l]&=&p.L^i[l]\\
                  &=&\widehat\algo(p^i, l)\\
                  &=&\widehat\algo(p^{i-1}, l)
    \end{eqnarray*}
    and $p$ cannot apply the rule $R_U$ in $\gamma^{i-1}$, a
    contradiction.
  \end{itemize}
  The lemma follows.  \qed
\end{proof}

\begin{lemma}\label{lem:round2D}
  If $\gamma^0$ is almost clean, then for
  $h\geq h_{\min(\durationIA, D)+1}$, $\gamma^h$ is clean.
\end{lemma}
\begin{proof}
  First, since $\gamma^0$ is almost clean,
  Lemma~\ref{lem:almostCleanClosed} implies that all $\gamma^h$ wth
  $h\geq 0$ are almost clean.

  The key element of the proof is that, in an almost clean
  configuration, what prevent a node $p$ in error from applying the
  rule $R_C$ is a neighbor $q$ also in error which is above $p$ (i.e.,
  $\hei{q}>\hei{p}$). Thus nodes in error of maximum height at the
  beginning of a round can apply the rule $R_C$, and thus, will have
  by the end of said round. This implies that, after each round, the
  maximum height of a node in error decreases by at least one. Since
  $\gamma^0$ is almost clean, the height of a node is at most
  $D$. And, by construction, it is also at most $\durationIA$.
  Therefore, $\gamma^{\min(\durationIA, D)+1}$ is clean, and by
  Lemma~\ref{lem:noRootCreation}, so are all configuration after. \qed
\end{proof}

\subsection{The ``algorithm phase''}

\begin{lemma}\label{lem:roundAlgoPhaseGreedy}
  Assume that $\gamma^0$ is clean.  In greedy mode, our algorithm
  either does not finish if $\durationIA=\infty$ or reaches a terminal
  configuration within at most $\durationIA$ rounds.
\end{lemma}
\begin{proof}
  The case $\durationIA$ follows from
  Lemma~\ref{lem:finalConfiguration}.  Otherwise, by
  Lemma~\ref{lem:almostCleanClosed}, all configurations are almost
  clean, so, no error rule is executed during $e$. Since $\gamma^0$
  contains no node in error, a node $p$ can only apply the rule $R_U$,
  which increases its height by one each time.

  Since any node $p$ with the lowest $\hei{p}<\durationIA$ can apply
  the rule $R_U$, the Lemma the follows from the fact that the minimum
  height of a node increases by at least one while the configuration
  is not terminal. \qed
\end{proof}

From now on, we assume that the algorithm runs in lazy mode.
Theorem~\ref{thm:stepComplexity} thus ensures that a terminal
configuration $\gamma^f$ exists.  By
Lemma~\ref{lem:finalConfiguration}, there exists $H$ such that for
each $p$, $\hei[f]{p}=H$.  We call $H$ the \emph{height} of the
terminal configuration.

Lemma~\ref{lem:stableConfiguration} implies that, $p.L[i]$ is the
state that the synchronous algorithm that we simulate assigns to $p$
at round $i$. Thus if $p.L[i+1]\neq p.L[i]$, it means that the
synchronous algorithm has not finished. Thus during our simulation, if
$p.L[i]\neq algo(p.L[i], N(p).L[i])$, then $p$ must apply the rule
$R_U$. This is the ``algorithm'' condition to apply the rule $R_U$.
The other condition is the ``catch up'' condition. If a node $p$ has a
neighbor $q$ such that $\hei{p}<\hei{q}$, then $p$ has to assume that
the synchronous algorithm has not finished, and therefore, it must
apply the rule $R_U$.  It is the catch up condition that ensures that
for all $p$, $\hei[f]{p}=H$.

We do not know how the lists are filled during the ``algorithm
phase''.  But any node $p$ such that $p.L[i]\neq p.L[i+1]$
($0\leq i<H$) may have been the first node to fill up the value
$p.L[i+1]$. We say that $p$ \emph{may have started line} $i+1$.  Now,
since $p.L[i+1]$ only depends on the values $q.L[i]$ for $q\in N[p]$,
if all $q\in N[p]$ are such that $q.L[i]=q.L[i+1]$, then no
$q\in N[p]$ may start line $i+2$. Therefore, if $p$ may start line
$i+1$, then either $i=0$ or there exists $q\in N[p]$ which may start
line $i$.

This motivates the following definition. A \emph{starting sequence}
for $\gamma^f$ is a sequence of nodes $s_1s_2\cdots s_H$ such that
each $s_i$ starts line $i$, and $s_{i-1}\in N[s_i]$ if $i>1$.  Note
that a terminal configuration may not have a starting sequence.
Indeed, if $\gamma^f$ is terminal and we set $p.L[H+1]:=p.L[H]$, then
the new configuration is also terminal but contains no starting
sequence.  Also, if $\gamma^f$ contains a starting sequence, then
$H=\Time$.

\begin{lemma}\label{lem:roundAlgoPhaseLazy1}
  If $\gamma^0$ is clean and $\gamma^f$ contains a starting sequence,
  then $e$ reaches a terminal configuration in at most $D+3\Time-2$
  rounds in lazy mode.
\end{lemma}
\begin{proof}
  According to the assumptions on $\gamma^0$, nodes can only apply
  Rule $R_U$ along the execution (Lemma~\ref{lem:almostCleanClosed}).
  This also implies that $\hei{p}$ can only increase.  Let
  $s_1\cdots s_\Time$ be a starting sequence of $\gamma^f$.  We also
  let $s_i = s_1$, for any $i < 1$.  For any node $p$ and
  $0\leq i\leq \Time$, we let $\lambda(p, i)=3i-2+d(p, s_i)$.  The
  lemma is a direct consequence of the following induction.

  We now prove by induction on $0\leq j\leq 3\Time+D-2$ that for every
  $p$ and $i$ such that $\lambda(p, i)\leq j$, we have $\hei{p}\geq i$
  (forever) from $\gamma^{h_j}$.

  If $j=0$, then $i=0$ and the result is clear.

  Suppose that $j> 0$.  If no $(p, i)$ such that $\lambda(p, i)=j$
  exists, then we are done.  Otherwise, let $(p, i)$ be such a pair.
  For any $q\in N[p]$,
  $\lambda(p, i)-\lambda(q, i-1)=3+d(p, s_i)-d(q, s_{i-1})$, and thus
  $\lambda(p, i)-\lambda(q, i-1)\geq 3-|d(p, s_i)-d(p,
  s_{i-1})|-|(d(p, s_{i-1})-d(q, s_{i-1})|$.  Now
  $|d(p, s_i)-d(p, s_{i-1})|\leq 1$ because $s_{i-1}\in N[s_i]$, and
  $|(d(p, s_{i-1})-d(q, s_{i-1})|\leq 1$ because $q\in N[p]$.  We thus
  have $\lambda(q, i-1)<j$.  By induction hypothesis, for any
  $q\in N[p]$, $\hei{q}\geq i-1$ in $\gamma^{h_{j-1}}$.

  Three cases now arise:
  \begin{itemize}
  \item If $\hei{p}\geq i$ in $\gamma^{h_{j-1}}$, then we are done.
  \item If $\hei{p}=i-1$ in $\gamma^{h_{j-1}}$ and $p=s_i$.  Then,
    since $s_i$ starts Line $i$, $p$ can apply the rule $R_U$ in
    $\gamma^{h_{j-1}}$, and thus will have done at last at
    $\gamma^{h_j}$.

  \item If $\hei{p}=i-1$ in $\gamma^{h_{j-1}}$ and $p\neq s_i$.  Then,
    let $q\in N(p)$ be such that $d(q, s_i)<d(p, s_i)$.  We have
    $\lambda(q, i)<j$, and thus, by induction hypothesis,
    $\hei{q}\geq i$ in $\gamma^{h_{j-1}}$.  This implies that $p$ can
    apply the rule $R_U$ in $\gamma^{h_{j-1}}$, and thus will have
    done at last at $\gamma^{h_j}$.  \qed
  \end{itemize}
\end{proof}

\begin{lemma}\label{lem:roundAlgoPhaseLazy2}
  If $\gamma^0$ is clean and $\gamma^f$ contains no starting sequence,
  then the execution $e$ reaches a terminal configuration within at
  most $2D$ rounds in lazy mode.
\end{lemma}
\begin{proof}
  According to the assumptions on $\gamma^0$, nodes can only apply
  Rule $R_U$ along the execution (Lemma~\ref{lem:almostCleanClosed}).
  This also implies that $\hei{p}$ can only increase.  Let $H$ be the
  height of $\gamma^f$.

  We first claim that there exists a node $s$ such that $\hei{s}=H$ in
  $\gamma_0$.  Indeed otherwise, since such a node exists in
  $\gamma^f$, choose the smallest $h\leq f$ such that $\hei[h]{s'}=H$
  for some node $s'$.  The node $s'$ starts line $H$, which implies
  the existence of a starting sequence, contradicting then our
  assumptions.

  We now let $\lambda(p, i)=2(i+D-H)-D+d(p, s)$.

  The rest of the proof is now very similar to the proof of the
  previous lemma.  We prove by induction on $0\leq j\leq 2D$ that if
  $\lambda(p, i)\leq j$, then in $\gamma^{h_j}$, $\hei{p}\geq i$.

  Suppose that $j=0$. We claim that if $\lambda(p, i)\leq 0$, then
  $i\leq H-d(p, s)$.  Indeed, for all $y$,
  $\lambda(p, y)<\lambda(p, y+1)$ and
  $\lambda(p, H-d(p, s))=D-d(p, s)\geq 0$.  To prove the base case, it
  is enough to prove that for all $p$, $\hei[0]{p}\geq H-d(p, s)$
  which follows from the fact that $\gamma^0$ contains no node in
  error and is almost clean.

  Suppose that $j>0$.  Let $(p, i)$ be such that $\lambda(p, i)=j$.
  For any $q\in N[p]$,
  $\lambda(p, i)-\lambda(q, i-1)=2+d(p, s)-d(q, s)>0$.  So
  $\lambda(q, i-1)<j$ and, by induction hypothesis, $\hei{q}\geq i-1$
  in $\gamma^{h_{j-1}}$.

  Two cases now arise:
  \begin{itemize}
  \item If $\hei{p}\geq i$ in $\gamma^{h_{j-1}}$, then we are done.
  \item If $\hei{p}=i-1$ in $\gamma^{h_{j-1}}$, then $p\neq s$.  Then
    let $q\in N(p)$ be such that $d(q, s)<d(p, s)$.  We have
    $\lambda(q, i)<j$, and thus, by induction hypothesis,
    $\hei{q}\geq i$ in $\gamma^{h_{j-1}}$.  This implies that $p$ can
    apply Rule $R_U$ in $\gamma^{h_{j-1}}$, and thus will have done at
    last at $\gamma^{h_j}$.  \qed
  \end{itemize}
\end{proof}

\subsection{The round complexity proof}

The following Theorem is a direct consequence of
Lemmas~\ref{lem:finalConfiguration}, \ref{lem:roundD},
\ref{lem:round2D}, \ref{lem:roundAlgoPhaseGreedy},
\ref{lem:roundAlgoPhaseLazy1} and \ref{lem:roundAlgoPhaseLazy2}.
\begin{theorem}
  Our algorithm reaches a clean configuration in at most
  $2+2\min(\durationIA, D)$ rounds. It does not end if
  $\durationIA=\infty$ in greedy mode, and it reaches a terminal
  configuration in at most
  \begin{itemize}
  \item $\min(2\durationIA, 2D)+2+\durationIA$ rounds in greedy mode
    when $\durationIA<\infty$.
  \item $\min(2\durationIA, 2D)+\max(2D+2, D+3\Time)$ round in lazy
    mode.
  \end{itemize}
\end{theorem}

\section{Instances}\label{instances}

We now develop several examples to illustrate the versatility and the
efficiency of our approach. All of them solve classical distributed
computing problems.

\subsection{Leader Election}\label{ex:leader}

\subsubsection{The Problem}

Recall that leader election requires all nodes to eventually
permanently designate a single node of the network as the leader. To
that goal, we assume an identified network and nodes will compute the
identifier of the leader.

\subsubsection{The Algorithm}
The network being identified, we do not need to distinguish
channels. So, \texttt{label} is the singleton $\{\bot\}$.

The state of each node $p$ includes its own identifier $ID$ (a
non-modifiable integer) and an integer variable $Best$ where $p$ will
store the identifier of the leader. This latter variable is
initialized with $p$'s own identifier. Hence, the predicate
\texttt{isValid} just checks that (1) no two nodes have the same
identifier and (2) all $Best$ variables are correctly initialized,
i.e., the initial value of each variable is equal the node identifier.

At each synchronous round, each node updates its variables $Best$ with
the minimum value among the $Best$ variables of its closed
neighborhood, and thus learns the minimum identifier of nodes one hop
further. The function \algo is defined accordingly; see
Algorithm~\ref{alg:leader}. After at most $D$ rounds, the $Best$
variable of each node is forever equal to the minimum identifier in
the network: the algorithm is eventually stable.

\begin{algorithm}[thbp]
  \small \SetKw{Input}{inputs:}

  \Input\\
  \begin{tabular}{lll}
    $\stateI{p}$ & : & \state, the state of $p$ \hfill\ {\em /*} initially, $\stateI{p}.Best = \stateI{p}.ID$  {\em */}\\
    {\tt NeigSet}$_p$ & : & set of pairs in \texttt{label} $\times$ \state  \hfill\ {\em /*} from the neighborhood {\em */}\\
  \end{tabular}
  \bigskip

  \Begin{
    Let $minID = \min(\{\stateI{p}.Best\} \cup \{s.Best\ |\ (\perp,s) \in {\tt NeigSet}_p\})$;\\
    {\bf return} $(\stateI{p}.ID,minID)$; }

  \caption{Function \algo of node $p$ for the leader election.}
  \label{alg:leader}
\end{algorithm}

\subsubsection{Contribution and Related Work}
Using our transformer in the lazy mode, we obtain a fully-polynomial
silent self-stabilizing leader election algorithm that stabilizes in
$O(D)$ rounds and $O(n^3)$ moves. Moreover, by giving an upper bound
$\durationIA$ on $D$ as input of the transformer, we obtain a
bounded-memory solution achieving similar time
complexities. Precisely, if we made the usual assumption that
identifiers are stored in $O(\log n)$ bits, we obtain a memory
requirement in $O(\durationIA.\log n)$ bits per node.

To our knowledge, our solution is the first fully-polynomial
asynchronous silent self-stabilizing solution of the
literature. Indeed, several self-stabilizing leader election
algorithms~\cite{DaLaVe11a,DaLaVe11b,AlCoDeDuPe17}, written in the
atomic-state model, have been proposed for arbitrary connected and
identified network assuming a distributed unfair daemon. However, none
of them is fully-polynomial. Actually, they all achieve a
stabilization time in $\Theta(n)$ rounds. Note that the algorithm
in~\cite{AlCoDeDuPe17} has a stabilization time in steps that is
polynomial in $n$, while~\cite{DaLaVe11a,DaLaVe11b} have been proven
to stabilize in a number of steps that is at least exponential in $n$;
see~\cite{AlCoDeDuPe17}.  Notice also that the algorithm proposed
in~\cite{KrKu13} actually achieves a leader election in $O(D)$ rounds,
however it assumes a synchronous scheduler.

\subsection{Breadth-First Search Spanning Tree Construction}

\subsubsection{The Problem}
We now consider the problem of distributedly computing a breadth-first
search (BFS) spanning tree in a rooted network. By ``distributedly'',
we mean that every non-root node will eventually designate the channel
toward its parent in the computed spanning tree. Being BFS, the length
of the unique path in the tree from any node $p$ to the root $r$
should be equal to the distance from $p$ to $r$ in the network.

This time, nodes are not assumed to be identified. Instead, we need to
distinguish channels using port numbers, for example.

\subsubsection{The Algorithm}
The state of each node $p$ contains a non-modifiable boolean $Root$
indicating whether or not the node is the root and a parent pointer
$Par$ that takes value in $\texttt{label} \cup \{NULL\}$. Initially,
each parent pointer is set to $NULL$. So, the predicate
\texttt{isValid} needs to check that (1) exactly one node has its
$Root$-variable equal to true, (2) each $Par$-variable is $NULL$, and
(3) locally, channels have distinct labels.

At each round, each non-root node $p$ whose $Par$-pointer is $NULL$
checks whether a neighbor is the root or has a non-$NULL$
$Par$-pointer; in this case $p$ (definitely) designates the channel to
such a neighbor with its pointer. If several neighbors satisfy the
condition, $p$ breaks ties using channel labels. The function \algo is
defined accordingly; see Algorithm~\ref{alg:bfs}.

\begin{algorithm}[thbp]
  \small \SetKw{Input}{inputs:}

  \Input\\
  \begin{tabular}{lll}
    $\stateI{p}$ & : & \state, the state of $p$ \hfill\ {\em /*} initially, $\stateI{p}.Par = NULL$  {\em */}\\
    {\tt NeigSet}$_p$ & : & set of pairs in \texttt{label} $\times$ \state  \hfill\ {\em /*} from the neighborhood {\em */}\\
  \end{tabular}
  \bigskip

  \Begin{ \uIf{$\stateI{p}.Root \vee \stateI{p}.Par \neq NULL$}{{\bf
        return $\stateI{p}$;} } \uElse{
      \uIf{$\exists (c,s) \in {\tt NeigSet}_p\ |\ s.Root \vee s.Par
        \neq NULL$}{ {\bf return} $(\stateI{p}.Root,c_{min})$, where
        $c_{min} = \min(\{c \ |\ (c,s) \in {\tt NeigSet}_p \wedge
        (s.Root \vee s.Par \neq NULL)\})$; } \uElse{{\bf return
          $\stateI{p}$};} } }

  \caption{Function \algo of node $p$ for the BFS spanning tree
    construction.}
  \label{alg:bfs}
\end{algorithm}

After at most $D$ synchronous rounds, all non-root nodes have a
parent, i.e., the BFS spanning tree is (definitely) defined and so the
algorithm is eventually stable.

\subsubsection{Contribution and Related Work}
Similarly to the leader election instance, using our transformer in
the lazy mode, we obtain a fully-polynomial silent self-stabilizing
leader election algorithm that stabilizes in $O(D)$ rounds and
$O(n^3)$ moves. Moreover, by giving an upper bound $\durationIA$ on
$D$ as input of the transformer, we obtain a bounded-memory solution
achieving similar time complexities. Precisely, its memory requirement
is $O(\durationIA.\log \Delta)$ bits per node, where $\Delta$ is the
maximum node degree in the network.

To our knowledge, our solution is the first fully-polynomial
asynchronous silent self-stabilizing solution of the literature that
achieves a stabilization time asymptotically linear in rounds.
Indeed, several self-stabilizing algorithms that construct BFS
spanning trees in arbitrary connected and rooted networks have been
proposed in the atomic-state
model~\cite{HuCh92,CoDeVi09,CoRoVi19,DaDeJoLa19}.  In~\cite{DeJo16},
the BFS spanning tree construction of Huang and Chen~\cite{HuCh92} is
shown to be exponential in steps.  The algorithm in~\cite{CoDeVi09} is
not silent and computes a BFS spanning tree in $O(\Delta \cdot n^3)$
steps and $O(D^2+n)$ rounds.  The silent algorithm given
in~\cite{CoRoVi19} has a stabilization time in~$O(D^2)$ rounds
and~$O(n^6)$ steps. The algorithm given in~\cite{Jo97} is not silent
and is shown to stabilize in $O(D\cdot n^2)$ rounds
in~\cite{DaDeJoLa19}, however notice that its memory requirement is in
$O(\log \Delta)$ bit per node.

Another self-stabilizing algorithm, implemented in the link-register
model, is given in~\cite{DoIsMo93}. It uses unbounded node local
memories.  However, it is shown in~\cite{DeJo16} that a
straightforward bounded-memory variant of this algorithm, working in
the atomic state model, achieves an asymptotically optimal
stabilization time in rounds, i.e., $O(D)$ rounds where $D$ is the
network diameter; however its step complexity is also shown to be at
least exponential in $n$.

\subsection{3-coloring in Rings}

\subsubsection{The Problem}

The coloring problem consists in assigning a color (a natural integer)
to every node in such a way that no two neighbors have the same color.
We now present an adaptation of algorithm of Cole and
Vishkin~\cite{CoVi86} that computes a 3-coloring in any oriented ring
of $n$ identified nodes.  The algorithm further assumes that node
identifiers are chosen in $[0..n^c-1]$, with $c \in \mathds N^*$.
Under such assumptions, the algorithm computes a vertex 3-coloring in
$log^*(n^c)+7$ rounds.

\subsubsection{The Algorithm}

The orientation of the ring is given by the channel labels. A node
should distinguish the state of its clockwise neighbor from its
counterclockwise one.  For instance, we can assume the channel number
of the clockwise neighbor is smaller than the one of counterclockwise
neighbor.  Without the loss of generality, we use two channel labels:
$L$ (for Left) and $R$ (for Right).  A consistent orientation is
obtained by assigning different labels for the two channels of each
node and different labels to the incoming and outgoing channels of
each edge.

The state of each node $p$ includes its own identifier $ID$ and four
variables:
\begin{enumerate}
\item $ph \in \{0,1\}$, initialized to 0; $ph$ indicates the current
  phase of the algorithm.
\item $maxColSize$, a natural integer initialized to
  $\lceil \log_2(n^c) \rceil$; $maxColSize$ is an upper bound on the
  number of bits necessary to store the current largest color.
\item $nbRds\in\{0,1,2,3\}$, initialized to 3, indicates the number of
  remaining rounds in Phase 1.
\item $color \in [0..n^c-1]$, initialized to $ID$, is the current
  color of the node.
\end{enumerate}

The predicate \texttt{isValid} should checks that (1) the ring
orientation is correct, (2) identifiers are taken in $[0..n^c-1]$ and
no two nodes have the same identifier, and (3) all variables are
correctly initialized.

\begin{algorithm}[thbp]
  \small \SetKw{Input}{inputs:} The type $\state$ is a record of five
  natural integers: $ID$, $ph$, $maxColSize$, $nbRds$, $color$.

  \bigskip

  \Input\\
  \begin{tabular}{lll}
    $\stateI{p}$ & : & \state, the state of $p$ \\
    \multicolumn{3}{r}{{\em /*} initially, $\stateI{p}.ph = 0$, $\stateI{p}.maxColSize= \lceil \log_2(n^c) \rceil$, $\stateI{p}.nbRds=3$ {\em */}}\\
    \multicolumn{3}{r}{/* $\stateI{p}.color=\stateI{p}.ID$; $\stateI{p}.ID$ is the node identifier {\em */}}\\
    {\tt NeigSet}$_p$ & : & set of pairs in \texttt{label} $\times$ \state  \hfill\ {\em /*} from the neighborhood {\em */}\\
  \end{tabular}

  \SetKw{Macros}{Macros:}

  \bigskip

  \Macros\\
  $bin(col)$ is $col$ interpret as a little-endian bit strings\;
  $col[i]$ is the value of $i$th bit of $col$ interpret as a
  little-endian bit strings\; $bitD (c1, c2)$ is the lowest index $i$
  such that $bin(c1)$ and $bin(c2)$ differ \; $posD(c1, c2)$ is
  $2\cdot bitD(c1 ,c2)+ c1[bitD(c1, c2)]$\;

  \bigskip

  \Begin{
   \uIf {$\stateI{p}.ph = 0$} { 
     Let $colS_p = s.color$ such that $(R,s) \in {\tt NeigSet}_p$\;
     $oldC_p := \stateI{p}.color$\;
     $oldMCS_p := \stateI{p}.maxColSize$\;
      \uIf {$1 + \lceil \log_2 (oldMCS_p) \rceil = oldMCS_p$}
      { 
      {\bf return} ($1$, $oldMCS_p$,  $\stateI{p}.nbRds$, $posD (oldC_p, colS_p)$)\;
      }
     \uElse{
   {\bf return} ($0$, $1 + \lceil \log_2 (oldMCS_p) \rceil$,  $\stateI{p}.nbRds$, $posD (oldC_p ,colS_p)$)\;
    }  
    }
   \uElse{
      \uIf {      $ \stateI{p}.nbRds > 0$}
      { 
      \uIf     {   $\stateI{p}.color = 2+\stateI{p}.nbRds$}
      {
        $nc$ :=  the first  color not in $\{st.color ~|~  (-,st) \in {\tt NeigSet}_p\}$\;
        {\bf return} ($1$, $\stateI{p}$.MaxColSize,  $\stateI{p}.nbRds-1$,  nc)\;
        }
        \uElse{
        {\bf return} ($1$, $\stateI{p}$.MaxColSize,  $\stateI{p}.nbRds- 1$,  $\stateI{p}$.color)\;
       }
   }
    \lElse {     {\bf return}  $\stateI{p}$}
   }

 }
 \caption{ Function \algo of node $p$ for 3-coloring on oriented
   rings}
 \label{alg:3-coloring}
\end{algorithm}

\algo is given in Algorithm~\ref{alg:3-coloring}.  At the end of an
arbitrary round of Phase 0, the correct coloring of the ring is
maintained, yet the upper bound on the number of bits necessary to
store the largest color reduces from $maxColSize$ to
$\lceil \log_2 maxColSize\rceil+1$.

The node $p$ detects that Phase 0 is over when the upper bound on the
number of bits necessary to store the largest color has not changed at
the end of the current round.  In this case, Phase 1 can start: each
node has a color in \{0, 1, 2, 3, 4, 5\}.

The three rounds of Phase 1 allow to remove colors 5, 4, and 3; in
that order.  In a round of Phase 1, any node having the color to
remove takes the first unused color in its neighborhood, this color
belongs to the set \{0, 1, 2\} as the network topology is a ring.
After $3$ rounds, all nodes has a color in \{0, 1, 2\}.

The synchronous algorithm terminated (no node changes its state) after
$log^*(n^c)+7$ rounds.

Using our transformer in the greedy mode with $B \geq log^*(n^c)+7$,
we obtain a silent self-stabilizing 3-coloring algorithm on oriented
rings that stabilizes in $O(B)$ rounds and $O(n^2B)$ moves.  Moreover,
its memory requirement is in $O(B.\log n)$ bits per node.  If we
carefully choose $B$ to be in $O(log^*(n))$, then we obtain a solution
that stabilizes in $O(\log^* n)$ rounds and $O(\log^*(n).n^2)$ moves
using $O(\log^* n.\log n)$ bits per node.

\subsubsection{Contribution and Related Work}
To our knowledge, our solution is the first self-stabilizing
3-coloring ring algorithm achieving such small complexities.

Indeed, self-stabilizing node coloring has been almost exclusively
investigated in the context of anonymous networks. Vertex coloring
cannot be deterministically solved in fully asynchronous
settings~\cite{An80}. This impossibility has been circumvented by
considering central schedulers or probabilistic
settings~\cite{GrTi00,BeDeGrPaTi10,BeDeGrTi09}.

In~\cite{LeSuWa09}, a self-stabilizing version of the Cole and Vishkin
algorithm is proposed, but the solution (based on the rollback of
Awerbuch and Varghese~\cite{AwVa91}) does not achieve a move
complexity polynomial in $n$.

\subsection{$k$-clustering}

\subsubsection{The Problem}
The clustering problem consists in partitioning the nodes into
clusters. Each cluster is a BFS tree rooted at a so-called
clusterhead. The $k$-clustering specialization of the problem requires
each node to be at a distance at most $k$ from the clusterhead of the
cluster it belongs to.

\subsubsection{The Algorithm}
A silent self-stabilizing algorithm that computes a clustering of at
most $\lceil \frac{n}{k+1}\rceil$ clusters in any identified and
rooted networks is proposed in~\cite{DaDeHeLaRi16}.  This algorithm is
a (hierarchical collateral~\cite{DaLaDeHeRi13}) composition of two
layers. The first one can be any silent self-stabilizing spanning tree
construction, and the second one is actually a silent self-stabilizing
algorithm for oriented tree networks.  The correctness of the
algorithm is established assuming a distributed weakly-fair
daemon. The stabilization time in rounds depends on the used spanning
tree constructions, indeed once the spanning tree is available, the
second layer stabilizes in at most $2H+3$ rounds, where $H$ is the
height of the tree.  So, using for example, the BFS spanning tree
construction given in~\cite{DeJo16}, we obtain a stabilization time on
$O(D)$ rounds.  However, since the correctness of the algorithm is
established under a daemon that is stronger than the distributed
unfair daemon, the move complexity of the solution cannot be bounded.

We now propose to make this algorithm fully-polynomial using our
transformer.

First, concerning the hypotheses, we only assume that nodes are
identified. We give up the root assumption, and we neither use any
labeling on channels.

Actually, we will modify the synchronous eventually stable leader
election of Subsection~\ref{ex:leader} to compute a BFS spanning tree
rooted at the leader node in at most $2D-1$ synchronous rounds and
then computes the $k$-clustering along the tree in at most $2D+3$
additional synchronous rounds.

To that goal, we need the following constants and variables:
\begin{itemize}

\item $ID$, the node identifier (a non-modifiable integer stored in
  $O(\log n)$ bits);

\item $Best$, an integer variable (stored in $O(\log n)$ bits)
  initialized to $ID$, this variable will eventually definitely
  contain the identifier of the leader;

\item $Dist$, an integer variable (stored in $O(\log D)$ bits)
  initialized to 0 where will be stored the distance between the node
  and the leader;

\item $Par$, a pointer (stored in $O(\log n)$ bits) designating an
  identifier and initialized to $ID$, this pointer aims at designating
  the parent of the node in the tree (the parent of the leader will be
  itself); and

\item the variables of the $k$-clustering for tree
  of~\cite{DaDeHeLaRi16} (stored in $O(\log k + \log n)$ bits). Those
  variables do not require initialization, since the algorithm
  of~\cite{DaDeHeLaRi16} is self-stabilizing.
\end{itemize}

The predicate \texttt{isValid} is defined similarly to previous
examples. We now outline the definition of \algo.

At each synchronous round, each node $p$ performs the following
actions in sequence:
\begin{enumerate}
\item $p$ update its $Best$ variable to be the minimum value among the
  $Best$ variables of its closed neighborhood.
\item If $Best$ is equal to the node identifier, then $(Dist,Par)$ is
  set to $(0,ID)$. Otherwise, $Dist$ is set to the minimum value among
  the $Dist$ variables of neighborhood plus one, and $Par$ is updated
  to designate the neighbor with the smallest $Dist$-value.
\item If necessary, $p$ updates its $k$-clustering variables according
  to the algorithm in~\cite{DaDeHeLaRi16}.

  Notice that for this latter action, we should define the following
  predicate and macro: (1) $Root(p)$ which is true if and only if the
  $Best$ variable of $p$ is equal to its own identifier, (2)
  $Children(p)$ that returns the children of $p$ in the tree, {\em
    i.e.}, the identifiers associated to neighboring states where the
  $Par$-variable designates $p$.
\end{enumerate}

Like in the leader election example, the $Best$ variable constantly
designates a unique leader after at most $D$ synchronous
rounds. Moreover, from that point, the $Dist$ and $Par$ variables of
the leader are forever equal to 0 and its identifier,
respectively. Hence, within at most $D-1$ additional synchronous
rounds, the values of all $Dist$ and $Par$ become constant and define
a BFS spanning tree. From that point, the $k$-clustering variables
stabilize in $2D+3$ additional synchronous rounds.  Overall, we obtain
the stability in at most $4D+2$ synchronous rounds.

\subsubsection{Contribution and Related Work}

Using our transformer in the lazy mode, we obtain a fully-polynomial
solution that constructs at most $\lceil \frac{n}{k+1}\rceil$ clusters
and whose stabilization times in rounds and moves are of the same
order of magnitude as the two previous examples, i.e., $O(D)$ rounds
and $O(n^3)$ moves. Moreover, by giving any value
$\durationIA \geq 4D+2$ as input of the transformer, we also obtain a
bounded-memory solution achieving similar time
complexities. Precisely, we obtain a memory requirement in
$O(\durationIA.(\log k+\log n))$ bits per node.

Again, several asynchronous silent self-stabilizing $k$-clustering
algorithms have been proposed in the literature
~\cite{DaLaDeHeRi13,Jo15,DaDeHeLaRi16,DaDeLa19}, but none of them
achieves full polynomiality. Actually, they offers various structural
guarantees such as minimality by inclusion~\cite{Jo15,DaDeLa19},
bounds on the number of clusters~\cite{DaLaDeHeRi13,DaDeHeLaRi16}, or
approximation ratio~\cite{DaDeHeLaRi16}.  The stabilization time in
rounds of~\cite{DaLaDeHeRi13,DaDeHeLaRi16,DaDeLa19} is in $O(n)$,
while the one of \cite{Jo15} is unknown. Moreover, to the best of our
knowledge none of those algorithms have a proven bound on its
stabilization time in moves.

It is worth noting that a particular spanning tree construction,
called {\em MIS tree}, is proposed in~\cite{DaDeHeLaRi16}. Actually, a
MIS tree is a spanning tree whose nodes of even level form a maximal
independent set of the network. Using this spanning tree construction
as the first layer of the $k$-clustering allows to obtain interesting
competitive ratios (related to the number of clusters) in unit disk
graphs and quasi-unit disk graphs. However, these desirable properties
are obtained are the price of augmenting the stabilization time in
rounds. Indeed, the MIS tree construction stabilizes in $\Theta(n)$
rounds and the exact problem solved by this construction is shown to
be $\mathcal P$-complete~\cite{DaDeHeLaRi16}.  The replacement of the
BFS construction by a MIS tree construction (left as an exercise)
would allow to obtain those approximation ratios in $O(n)$ rounds and
a number of moves polynomial in $n$.

\section{The Rollback Compiler}\label{sect:roll}

Awerbuch and Varghese have proposed a transformer called
\emph{Rollback} to self-stabilize synchronous algorithms dedicated to
static tasks. Consider a synchronous algorithm for some static task
that terminates in $T$ rounds. The basic principle of the transformer
consists for each node to keep a log of the states it took along the
$T$ rounds. Assume each node stores its log into the array $p.t$ where
\begin{itemize}
\item $p.t[0]$ is its initial state ($p.t[0]$ is read-only, and so
  cannot be corrupted) and
\item $p.t[i]$, with $0< i \leq T$, is the state computed by $p$
  during the $i^{th}$ round.
\end{itemize}
Then, every node $p$ just has to continuously checked and corrected
$p.t[i]$, for every $i\in [1..T]$: every $p.t[i]$ should be the state
obtained by applying the local algorithm of $p$ on the states
$q.t[i-1]$ of every $q \in N[p]$ (its closed neighborhood).

Consider now the following synchronous algorithm, $A$. The state of
every node $p$ in $A$ is made of a constant input $p.I \in \{0,1\}$
and a variable $p.S \in \{0,1\}$. The variable $p.S$ is initialized to
$p.I$. Then, the local algorithm of $p$, $A(p)$, consists of a single
rule:
\begin{center}
  $R_{\min} : p.S \neq \min_{q \in N[p]}\{q.S\} \to p.S := \min_{q \in
    N[p]}\{q.S\}$
\end{center}
At each synchronous round, the rule $R_{\min}$ allows any node $p$ to
compute in $p.S$ the minimum $S$ value in its closed neighborhood.  In
the worst case (e.g., all inputs except one are equal to 1),
$\mathcal D$ rounds are required so that the execution of $A$ in a
network $G$ of diameter $\mathcal D$ reaches a terminal configuration.
Overall, $A$ computes the minimum value among all the boolean inputs,
i.e., $A$ reaches in $\mathcal D$ rounds a terminal configuration
where $p.S=\min_{q \in V(G)}\{q.I\}$.

We now study the step complexity of the self-stabilization version of
$A$ obtained with the rollback compiler; we denote the transformed
algorithm by $RC(A)$. For fair comparison, we assume the fastest
method to correct a node state: when activated, in one atomic step the
node recomputes all cells of its array in increasing order according
to the local configuration of its closed neighborhood.

Let $G_1$ be the path $b_1, a_1, c_1, d_1, e_1$. For every $x > 1$, we
construct $G_x$ by linking to $G_{x-1}$ the path
$b_x, a_x, c_x, d_x, e_x$ as follows: $b_x$ and $e_x$ are linked to
$c_{x-1}$. Figure~\ref{fig:1} shows the network $G_3$. In the
following, we denote by $V_{x}$, with $x \geq 1$, the subset of nodes
$\bigcup_{i \in [1..x]} \{b_i, a_i, c_i, d_i, e_i\}$, i.e., the set of
nodes of $G_x$.

Notice that for every $x > 1$, $G_x$ contains $5x$ nodes (i.e.,
$|V_{x}| = 5x$) and its diameter is $3x-1$. Hence, executing $RC(A)$
on $G_x$ with $x > 1$ requires that each node $p$ maintains an array
$p.t$ of $3x$ cells indexed from 0 to $3x-1$.

Given any positive number $d$, we denote by $\bar i$, with
$0 \leq i \leq d$, any array $t$ of size $d$ such that $t[j]=1$ for
every $j \in [0..i-1]$ and $t[j]=0$ for every $j \in [i..d-1]$.  The
\emph{index} of an array $\bar i$ is $i$.  In the following, we simply
refer to the index of the array of some node $p$ as \emph{the index of
  $p$}.

\begin{remark}
  Assume a configuration where arrays have size $d$ and every node $p$
  satisfies $p.t=\overline {i_p}$ for some value $i_p$.  When
  activating some node $q$ such that $i_q > 0$, if $0 < i^{\min} < d$
  is the minimum index in the closed neighborhood of $q$ ($N[q]$),
  then $q$ sets $q.v$ to $\overline{i^{\min}+1}$.
\end{remark}

\begin{center}
  \begin{figure}[htpb]
    \centering \includegraphics{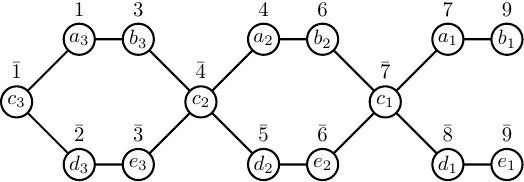}
    \caption{The network $G_3$}
    \label{fig:1}
  \end{figure}
\end{center}

 \begin{lemma}
   For every integer $x > 1$, there is an execution of $RC(A)$ on
   $G_x$ that requires at least $2^{x}-1$ steps to reach a terminal
   configuration.
 \end{lemma}
 \begin{proof}
   Assume the following initial configuration $\gamma_\text{init}$:
   \begin{itemize}
   \item All nodes have 1 as input.
   \item For every $i \in [1..x]$,
     \begin{itemize}
     \item $a_i.t=c_i.t=\overline{3(x-i)+1}$,
     \item $d_i.t=\overline{3(x-i)+2}$, and
     \item $b_i.t=e_i.t=\overline{3(x-i)+3}$.
     \end{itemize}
   \end{itemize}
   Notice that the configuration is well-defined since all nodes has
   input 1 and a positive index; moreover the maximum index is
   $3(x-1)+3 = 3x$.  (An example of possible initial configuration is
   given in Figure~\ref{fig:1}.)

   We now show by induction on $i$ that for every $i \in [1..x]$,
   there is a prefix $P_i$ of execution starting from
   $\gamma_\text{init}$ and containing at least $2^i-1$ steps such
   that
   \begin{itemize}
   \item no node in $V(G_x) \setminus V_{i} \cup \{c_i\}$ moves in
     $P_i$ and

   \item in the last configuration of $P_i$, noted $\gamma_{P_{i}}$,
     the index of all nodes $a_j$ with $j \in [1..i]$ have increased
     by 1 and no other index has changed.
   \end{itemize}

   For $i = 1$, activating $a_1$ gives a prefix satisfying all
   requirements.  Assume now $1< i \leq x$. We now explain how to
   expand the prefix $P_{i-1}$ given by the induction hypothesis.
   $P_{i-1}$ contains at least $2^{i-1}-1$ steps. Moreover,
   \begin{itemize}
   \item no node in $V(G_x) \setminus V_{i-1} \cup \{c_{i-1}\}$ has
     ever moved in $P_{i-1}$ and

   \item the index of all nodes $a_j$ with $j \in [1..i-1]$ in
     $\gamma_{P_{i-1}}$ have increased by 1 and no other index has
     changed.
   \end{itemize}
   So, we should extend $P_{i-1}$ to reach at least $2^i-1$ steps
   while never activating a node in
   $V(G_x) \setminus V_{i} \cup \{c_i\}$, moreover compared to
   $\gamma_{P_{i-1}}$, we should only increase the index of $a_i$ by
   one.

   Consider the path
   $P=b_ic_{i-1}d_{i-1}e_{i-1}\cdots c_2d_2e_2c_1d_1e_1$ and the
   following scheduling of activation starting from
   $\gamma_{P_{i-1}}$.

   \begin{enumerate}
   \item Because the difference of indexes in $b_{i}$ and $a_{i}$ is
     2, when activating $b_i$, its index decreases by 1.  But then the
     difference of indexes in $c_{i-1}$ and $b_{i}$ becomes 2 which
     allows us to decrease the index of $c_{i-1}$.  By activation the
     vertices of $P$ in order, all the indexes of the vertices of $P$
     decrease by 1.
   \item We can next activate all the vertices $a_j$ with $1\leq j< i$
     which also decrease their index by 1.
   \item We increase by 1 the index of $a_{i}$ by activating it.
   \item We now activate all nodes of $P$, in order, to increase their
     index by 1.
   \end{enumerate}
   Let $\gamma_\text{half}$ be the configuration reached after the
   previous scheduling. Compared to $\gamma_\text{init}$, only one
   index has changed: $a_i$. So, we re-apply the scheduling that has
   produced $P_{i-1}$. Let $\gamma_\text{final}$ be the reached
   configuration after applying this scheduling on
   $\gamma_\text{half}$. Compared to $\gamma_\text{init}$, the index
   of all nodes $a_j$ with $j \in [1..i]$ have increased by 1 and no
   other index has changed. Moreover, no node in
   $V(G_x) \setminus V_{i} \cup \{c_i\}$ has moved in the prefix that
   led to $\gamma_\text{final}$. Finally, that prefix has length at
   least $(2^{i-1}-1)+1+(2^{i-1}-1) = 2^i-1$. Thus, we can let
   $\gamma_{P_{i}} = \gamma_\text{final}$ and let $P_i$ be the prefix
   from $\gamma_\text{init}$ to $\gamma_{P_{i}}$ we have just
   exhibited. $P_i$ satisfies all requirements for $i$ and the
   induction holds.  Finally, by letting $i = x$, the lemma
   holds. \qed
 \end{proof}

\begin{corollary}
  For every integer $x > 1$, there is an execution of $RC(A)$ on $G_x$
  that requires at least $2^{\frac{n}{5}}$ steps
  (resp. $2^{\frac{\mathcal D+1}{3}}$) to reach a terminal
  configuration, where $n$ is the number of nodes in (resp.
  $\mathcal D$ is the diameter of) $G_x$.
\end{corollary}

Assume the known upper bound on the time complexity of $A$ is not
tight.  We can let $x$ to obtain a network $G_x$ and arrays of with
$T_x+1$ cells where $T_x$ is the known bound on the time complexity of
$A$ on $G_x$. Then, we can add a node $z$ linked to all other nodes so
that the diameter becomes 2. Finally, by considering a configuration
identical to the previous construction, except that $z$ has a maximum
index (i.e., $3x$), we can built the same execution (the presence of
$z$ has no impact due its index). Hence, we can obtain an exponential
lower bound that does not depend on the actual diameter but rather on
the known upper bound.

\section{Conclusion}\label{ccl}

We have proposed a versatile transformer that builds efficient silent
self-stabilizing solutions. Precisely, our transformer achieves a good
trade-off between time and workload since it allows to obtain
fully-polynomial solutions with round complexities asymptotically
linear in $D$ or even better.

Our transformer can be seen as a powerful tool to simplify the design
of asynchronous self-stabilizing algorithms since it reduces the
initial problem to the implementation of an algorithm just working in
synchronous settings.  By the way, all tasks, even
non-self-stabilizing ones, that terminate in synchronous settings can
be made self-stabilizing using our transformer, regardless the model
in which they are written (atomic-state model, Local model, \ldots).

Interestingly, the Local model which was initially devoted to prove
lower bounds, becomes an upper bound provider since we can extensively
use it to give inputs to our transformer that will, based on them,
construct asymptotically time-optimal self-stabilizing solutions.

Another interesting application of our transformer is the weakening of
fairness assumptions of silent self-stabilizing algorithms (e.g.,
asynchronous algorithms assuming a distributed weakly fair or a
synchronous daemon) without compromising efficiency (actually, such
algorithms can be provided as input of the transformer).

The perspectives of this work concern the space overhead and the moves
complexity. The space overhead of our solution depends on the
synchronous execution time of the input algorithm. In the spirit of
the resynchronizer proposed by Awerbuch and Varghese~\cite{AwVa91}, we
may build another space-efficient transformer that would assume more
constraint on input algorithms. Concerning the move complexity, for
many problems, the trivial lower bound in moves for the asynchronous
(silent) self-stabilization is $\Omega(D\times n)$; while we usually
obtain upper bounds in $O(n^3)$ moves with our transformer. Reduce the
gap between those two bounds is another challenging perspective of our
work.


\end{document}